\documentclass[twocolumn,aps,pra,groupedaddress,superscriptaddress,amsmath,amssymb,showpacs,noeprint]{revtex4-2}

\usepackage{graphicx}
\usepackage{dcolumn}% Align table columns on decimal point
\usepackage{bm}% bold math
\usepackage[dvipsnames]{xcolor}
\usepackage{hyperref}% add hypertext capabilities
\hypersetup{colorlinks,linkcolor={MidnightBlue},citecolor={MidnightBlue},urlcolor={MidnightBlue}}  
\usepackage{physics}
\usepackage{amsfonts}
\usepackage{verbatim}
\usepackage{amsmath}
\usepackage{csquotes}
\usepackage{color}
\usepackage{soul}
\usepackage{amsthm}
\usepackage{thmtools}
\newcommand{\prt}[1]{\left(#1\right)}

\newcommand{\prtq}[1]{\left[#1\right]}

\newcommand{\prtg}[1]{\left\{#1\right\} }
\newcommand{\prtb}[1]{\bigg(#1 \bigg)}

\newcommand{\prtbg}[1]{\bigg\{#1 \bigg\}}

\newcommand{\betag}{\beta_g}

\newcommand{\Id}{\mathbb{I}}

\newcommand{\Hsi}{H^{\textup{si}}}
\newcommand{\Hsim}{\tilde{H}^{\textup{si}}_I}

\newcommand{\Eexp}{E_{\textup{exp}}}

\declaretheorem{theorem}

\newenvironment{equations}
{\begin{equation}\begin{aligned}}
{\end{aligned}\end{equation}\ignorespacesafterend}

\newcommand{\stMin}{\ket{\psi_g}}
\newcommand{\Zmin}{Z_g}
\newcommand{\Emin}{E_g}

\newcommand{\E}{\mathcal{E}}

\newcommand{\loc}{MSSG}
\newcommand{\nonloc}{global}

\newcommand{\blue}{}

\newcommand{\bl}{l}

\begin{document}

%\preprint{APS/123-QED}

\title{Generation of minimum energy entangled states}

\author{Nicol\`o Piccione}
\email{nicolo.piccione@univ-fcomte.fr}
\affiliation{Institut UTINAM, CNRS UMR 6213, Universit\'{e} Bourgogne Franche-Comt\'{e}, Observatoire des Sciences de l'Univers THETA, 41 bis avenue de l'Observatoire, F-25010 Besan\c{c}on, France}

\author{Benedetto Militello}
\affiliation{Universit\`a degli Studi di Palermo, Dipartimento di Fisica e Chimica - Emilio Segr\`{e}, via Archirafi 36, I-90123 Palermo, Italy}
\affiliation{INFN Sezione di Catania, via Santa Sofia 64, I-95123 Catania, Italy}

\author{Anna Napoli}
\affiliation{Universit\`a degli Studi di Palermo, Dipartimento di Fisica e Chimica - Emilio Segr\`{e}, via Archirafi 36, I-90123 Palermo, Italy}
\affiliation{INFN Sezione di Catania, via Santa Sofia 64, I-95123 Catania, Italy}

\author{Bruno Bellomo}
\affiliation{Institut UTINAM, CNRS UMR 6213, Universit\'{e} Bourgogne Franche-Comt\'{e}, Observatoire des Sciences de l'Univers THETA, 41 bis avenue de l'Observatoire, F-25010 Besan\c{c}on, France}

\begin{abstract}	
Quantum technologies exploiting bipartite entanglement could be made more efficient by using states having the minimum amount of energy for a given entanglement degree. Here, we study how to generate these states in the case of a bipartite system of arbitrary finite dimension either by applying a unitary transformation to its ground state or through a zero-temperature thermalization protocol based on turning on and off a suitable interaction term between the subsystems.
In particular, we explicitly identify three possible unitary operators and five possible interaction terms. On one hand, two of the three unitary transformations turn out to be easily decomposable in terms of local elementary operations and a single \blue{nonlocal} one, making their implementation easier. On the other hand, since the thermalization procedures can be easily adapted to generate many different states, we numerically show that, for each degree of entanglement, generating minimum\blue{-}energy entangled states \blue{costs, in general,} less than generating the vast majority of the other states.
\end{abstract}

\maketitle

\section{Introduction}

Entanglement, apart from playing a crucial conceptual role in quantum mechanics~\cite{Einstein1935,Bell1964}, is widely considered a resource for quantum technologies since it enables various genuinely quantum protocols~\cite{Horodecki2009,BookNielsen2010}. Accordingly, the study of suitable protocols to generate entangled states in a great variety of physical setups has attracted much attention in \blue{recent} decades~\cite{Plenio1999,Braun2002,Bellomo2008,Sarlette2011, Krauter2011,Bellomo2013,Bellomo2015,Bellomo2017,Luo2017,Cakmak2019,Goldberg2019,Bellomo2019,Egger2019,Katz2020}.
However, generating quantum states with a certain amount of entanglement usually requires converting other resources into it, such as coherence~\cite{Xi2019,Korzekwa2019}, \blue{non}equilibrium thermal resources~\cite{Korzekwa2019}, and energy~\cite{Piccione2020Energy}.
Among these resources, energy is one of the most important with regard to the evaluation of the generation cost. Recently, some studies dealt with the energy cost of quantum operations~\cite{Ikonen2017}, including the generation of correlations~\cite{Bakhshinezhad2019} and entanglement~\cite{Galve2009,Beny2018,Hackl2019,Piccione2020Energy}. In particular, in Ref.~\cite{Piccione2020Energy} we found the bounds on the local energy of a bipartite system of arbitrary finite dimension imposed by having a certain amount of entanglement. Moreover, we have also identified the states saturating these bounds. We call any state saturating the lower bound a \enquote{minimum energy entangled state} (MEES).

Protocols to generate entangled states typically exploit unitary processes~\cite{Egger2019,Cakmak2019,Goldberg2019} and, possibly, measurements~\cite{ Bellomo2019,Cakmak2019,Katz2020}. On the other hand, dissipative processes have been also identified as possible sources of entangled states, e.g, when the steady states of the dissipative dynamics are entangled \cite{Braun2002,Bellomo2008,Sarlette2011,Krauter2011,Bellomo2013,Bellomo2015}. In particular, entangled states can be obtained through \blue{zero-temperature} thermalization protocols (such as those studied in Ref.~\cite{Piccione2019Work}) requiring the implementation of suitable interaction terms such that the ground state of the total Hamiltonian is the desired entangled state.

Although the majority of quantum information protocols rely on two-level systems~\cite{BookNielsen2010}, $d$-level systems (qudits) may be more powerful for information processing~\cite{Bechmann-Pasquinucci2000,Bullock2005,Lanyon2009,Yan2019,Kiktenko2020}.
Logical operations for qudits have already been implemented in systems like molecular magnet transistors~\cite{Godfrin2017}, superconducting systems~\cite{Blok2020}, and integrated optics~\cite{Kues2017FULL,Reimer2019,Imany2019}.  In particular, single-qudit gates connected to the generalized CNOT gate~\cite{Wilmott2011}, also called controlled-SUM~\cite{Blok2020}, have already been experimentally implemented~\cite{Gao2019,Isdrailua2019}. Moreover, universal quantum computation based on trapped ions qudits has been suggested to be feasible~\cite{Low2020}. In this context, efficient ways to generate MEESs of arbitrary dimensions are desirable.

In this paper, we present various proposals to generate MEESs efficiently, through unitary or dissipative dynamics. Regarding the unitary approaches, we provide explicit unitary operators connecting the separable ground state of the bipartite system to the desired MEES, whatever the dimensions of the subsystems are.
Remarkably, in the two approaches that we will denote as \blue{mostly single-system gates (\loc{})} ones, these unitary operators can be decomposed in elementary two-level rotation gates and a local change of phase acting on only one subsystem, and, as last gate, a generalized CNOT one, possibly making their implementation easier. Regarding the dissipative approaches, we find five classes of suitable Hamiltonians for zero-temperature thermalization and compare them.
For all classes, we show that the energetic cost generation for the MEESs is practically the minimum one.
Our calculations also show that, for generating MEESs, some interaction Hamiltonians always perform better than others.
Even if our aim is to generate MEESs, our proposals can be easily adapted to generate a large class of states. In particular, both for the unitary and dissipative scenarios, some proposals are specific for states having a structure similar to that of MEESs, while the other proposals could work for any target state.

The paper is organized as follows. In Sec.~\ref{sec:BackgroundInformation} we recall the definition and properties of MEESs, while in Sec.~\ref{sec:UnitaryTransformation} the unitary transformation approach to the generation of these states is discussed. In Sec.~\ref{sec:TwoApproaches}, a different approach based on zero-temperature thermalization processes is explored. In Sec.~\ref{sec:ApplTherm} we present a detailed comparison of the various methods based on thermalization focusing on a $3\times 4$ system. Finally, in Sec.~\ref{sec:Conclusions} we give some conclusive remarks, while we provide details of our analysis and explicit calculations in several \blue{a}ppendices.

\section{\label{sec:BackgroundInformation}Minimal background on minimum energy entangled states}

Let us consider an arbitrary finite bipartite system $S$ composed of subsystems $A$ and $B$, whose dimensions $N_A$ and $N_B$ satisfy $N_A \leq N_B$.
Their Hamiltonians, $H_A$ and $H_B$, have the following structure:
\begin{equation}
H_X = \sum_{j =0}^{N_X-1} X_j \dyad{X_j},\quad X=A,B,
\end{equation}
where $X_0 \leq X_1 \leq \cdots \leq X_{N_X-1}$. The free Hamiltonian of the bipartite system is then $H_0 = H_A + H_B$ and the dimension of the whole Hilbert space is $N_S=N_A N_B$.

In the case of pure states, the standard entanglement quantifier is the entropy of entanglement~\cite{Vidal2000,Plenio2007,Horodecki2009}, that is, the Von Neumann entropy of one of the reduced states\blue{,}
\begin{equation}
\E \prt{\ket{\psi}} = S\prt{\Tr_{A(B)}{\dyad{\psi}}},
\end{equation}
where $S(\rho) = - \Tr{\rho \ln \rho}$.
For mixed states various quantifiers exist, but a standard requirement is that they coincide with the entropy of entanglement when applied to pure states~\cite{Plenio2007}.

In Ref.~\cite{Piccione2020Energy} we showed that, for a given degree of entanglement $\E$, \blue{with the condition $\E > \ln d_g$, where $d_g$ is the minimum between the ground-energy degeneracy of the Hamiltonians $H_A$ and $H_B$,} the states minimizing the average of $H_0$ are pure states of the form 
\begin{equation}
\label{eq:MinimumState}
\stMin=\frac{1}{\sqrt{\Zmin}}\sum_{j=0}^{N_A-1}e^{i \theta_j} e^{-\blue{\prt{\betag/2}} E_j} \ket{E_j},
\end{equation}
where $1/\sqrt{\Zmin}$ is a normalization factor, $E_j = A_j + B_j$, $\ket{E_j} = \ket{A_j B_j}$, $e^{i \theta_j}$ are arbitrary phase factors, and $\betag$ is the positive solution of the equation
\begin{equation}
\label{eq:BetagEquation}
\prt{-\betag \pdv{\betag} + 1}\ln \Zmin = \E.
\end{equation}
%The energy of these states can be calculated as $\Emin = - \pdvbg \ln{\Zmin}$.
The energy of the MEESs $\stMin$ can be easily calculated as 
\begin{equation}
\label{eq:minimumenergy}
	\Emin = -\partial_{\beta_g} \ln Z_g.
\end{equation}
It is worth mentioning that the form of the MEESs is such that the parameter $Z_g$ is equal to the partition function of a thermal state. This is one of the properties of the MEESs. Among these properties, we also cite that MEESs are connected through local operations and classical communication, they naturally appear in some many-body systems, and they can be used in various ways to improve the energetic efficiencies of quantum protocols relying on entanglement. More details about all these facts are given in Ref.~\cite{Piccione2020Energy}.
\blue{We finally observe that in the trivial case $\E \leq \ln d_g$ the minimum energy is $E_0$ and minimum-energy pure states can be searched in the ground-energy eigenspace of $H_A + H_B$.}

The aim of this paper is to present possible strategies to generate states of the form of Eq.~\eqref{eq:MinimumState}. In particular, we will focus on generation protocols based, respectively, on the application of unitary transformations to the \blue{non}degenerate ground state $\ket{E_0}$ (hereafter we assume that $A_1 > A_0$ and $B_1 > B_0$), and on the dissipative dynamics associated to zero-temperature thermalizations.
The state $\ket{E_0}$ is the most natural choice for the starting state since, in many cases, it can be easily obtained as the result of an approximate zero-temperature thermalization. Indeed, if the energy gap between the ground and the first excited level of the Hamiltonian $H_A + H_B$ is much higher than $k_B T$, where $k_B$ is the Boltzmann constant and $T$ is the temperature of the bath, the steady state of the dynamics, i.e., the thermal state at temperature $T$ is a very good approximation of the ground state.

\section{\label{sec:UnitaryTransformation} The unitary transformation approach}

In this section, we provide three explicitly constructed unitary operators that can be easily decomposed into other ones.
Despite the fact that these three operators share a similar mathematical structure, they strongly differ in their physical implementation because the first one \blue{$U_S$} can be decomposed as a product of \blue{non}local simple operators \blue{whereas} the other two \blue{$\tilde{U}_A$ and $\tilde{U}_B$} are compositions of simple local unitary operators and a single generalized CNOT gate~\cite{Wilmott2011} applied as last operation.

\subsection{A \nonloc{} unitary transformation}

We obtain the operator $U_S$ as a change of basis between the eigenbasis of the Hamiltonian $H_0$ and a new basis obtained by applying the Gram-Schmidt orthogonalization process to the set of linearly independent vectors $\prtg{\ket{\phi},\ket{E_1},\dots,\ket{E_{N_A-1}}}$, where $\ket{\phi}$ has the form
\begin{equation}
\label{eq:TargetStates}
\ket{\phi} = \sum_{j=0}^{N_A -1} e^{i \theta_j} \sqrt{\lambda_j} \ket{E_j}.
\end{equation}
The rest of the basis, i.e., all the vectors $\ket{A_i B_j}$ with $i \neq j$, are left unchanged. By varying the state $\ket{\phi}$, a family of unitary operators is obtained.
This procedure, presented in detail in Appendix~\ref{sec:APPUnitaryOperatorConstruction} for a more general case, here leads to the vectors
\begin{multline}
\label{eq:CanonicalBasis}
\ket{\phi_k^S} =
\prtb{\gamma_{k}\ket{E_k} -e^{i(\theta_0 -\theta_k)}\sqrt{\lambda_0 \lambda_k} \ket{E_0}\\ - \sum_{j=k+1}^{N_A-1} e^{i(\theta_j -\theta_k)} \sqrt{\lambda_k \lambda_j} \ket{E_j}}/\sqrt{\gamma_k \gamma_{k-1}},
\end{multline}
where $1 \leq k \leq N_A-1$, $\gamma_k = 1 - \sum_{j=1}^{k} \lambda_j$, and $\gamma_0 = 1$.

Denoting $\ket{\phi_0^S}=\ket{\phi}$, the unitary operator $U_S$ takes the form
\begin{equation}
\label{eq:UnitaryOperatorGlobal}
U_S = \sum_{k=0}^{N_A-1} \dyad{\phi_k^S}{E_k} +\sum_{i=0}^{N_A-1} \sum_{j=0,j\neq i}^{N_B-1} \dyad{A_i B_j}.
\end{equation}
By choosing $\ket{\phi}=\stMin$, this family of operators  provides an explicit construction for an operator making the $\ket{E_0}\rightarrow \stMin$ transition.

Interestingly, the operator $U_S$ can be decomposed as the product of simple \blue{non}local operators (see Appendix~\ref{sec:APPUnitaryOperatorConstruction} for the proof)
\begin{equation}
U_S = \prod_{j=1}^{N_A} U_{S,j},
\end{equation}
where, for $1\le j <N_A$,
\begin{multline}
U_{S,j} = \sqrt{\frac{\gamma_j}{\gamma_{j-1}}} \prt{\dyad{E_0}+\dyad{E_j}} +\Id_{\perp,j} \\
+\sqrt{\frac{\lambda_j}{\gamma_{j-1}}} \prt{e^{i \theta_j}\dyad{E_j}{E_0}-e^{- i \theta_j}\dyad{E_0}{E_j}},
\end{multline}
while, for $j = N_A$,
\begin{equation}
U_{S,N_A} = e^{i \theta_0} \dyad{E_0} + \Id - \dyad{E_0} ,
\end{equation}
being 
\begin{equation}
\blue{\Id_{\perp,j} = \sum_{i=1, i \neq j} \dyad{E_i} +\sum_{i=0}^{N_A-1} \sum_{k=0, k \neq i}^{N_B-1} \dyad{A_i B_k}}
\end{equation}
the identity operator on the subspace perpendicular to the one spanned by $\ket{E_0}$ and $\ket{E_j}$, and $\Id$ the identity operator in the whole Hilbert space.
From a practical point of view, this may not be the best way to decompose  the operator $U_S$, as it implies the application of $N_A-1$ two-qudit gates, followed by a phase change.
However, a similar decomposition is very fruitful for the \loc{} unitary operators described in the next subsection. Since each component of the above decomposed form of the operator $U_S$  acts on the whole Hilbert space of system $S$, we refer to it as the \nonloc{} unitary operator, as opposed to the \loc{} ones.

Even if in this paper we concentrate on the generation of MEESs for a fixed entanglement, we remark that by applying the Gram-Schmidt procedure to the set $\{\ket{\varphi},\ket{A_0 B_1},\ket{A_0B_2},\dots,\ket{A_{N_A-1} B_{N_B-1}}\}$, the resulting
family of operators can describe any transition from $\ket{E_0}$ to almost any state $\ket{\psi}$ of the Hilbert space of system $S$. A solvable exception is the case when the target state is orthogonal to the initial state $\ket{E_0}$~\footnote{\label{note1}The states which have zero projection on $\ket{E_0}$ can indeed be obtained by means of simple modifications of the method. See the relevant discussion in Appendix~\ref{sec:APPUnitaryOperatorConstruction}.}.
Anyway, the simple form of the MEESs allow the unitary operators to be simpler when connecting $\ket{E_0}$ to one of them compared to when the target state is arbitrary. 
Indeed, the Schmidt decomposition of a MEES involves only the states $\ket{E_j}$, as in Eq.~\eqref{eq:TargetStates}, where every coefficient $\lambda_j$ is a simple function of the same parameter $\beta_g$. Moreover, the target MEES can be chosen with all the phases equal to zero.
In general, in the case of an arbitrary target state the resulting unitary operator is represented by a $N_S \times N_S$ matrix of \blue{non}trivial elements as opposed to the smaller number of \blue{non}trivial $N_A \times N_A$ elements of the matrix associated to Eq.~\eqref{eq:UnitaryOperatorGlobal}  which holds for target states like those of Eq.~\eqref{eq:TargetStates}.
The possibility to generate practically every state will be used in Sec.~\ref{sec:ApplTherm} to compare the behavior of the MEESs to all the other states in the zero-temperature thermalization approach.

\subsection{Mostly single-system gates (MSSG) unitary transformations\label{subsec:LocalUnitaryApproach}}

Another unitary transformation mapping $\ket{E_0}$ to an arbitrary state $\ket{\phi}$ of the form of Eq.~\eqref{eq:TargetStates} is the one composed as follows (as a particular case, $\stMin$ is obtained  when $  \sqrt{\lambda_j}=e^{-\frac{\betag}{2} E_j}/ \sqrt{Z_g}$). A local unitary operator $U_A$ acts as the previous operator $U_S$ of Eq.~\eqref{eq:UnitaryOperatorGlobal} but on system $A$, i.e., 
\begin{equation}
\label{eq:UnitaryOperatorLocalA}
U_A = \sum_{k=0}^{N_A-1} \dyad{\phi_k^A}{A_k},
\end{equation}
where each $\ket{\phi_k^A}$ is equal to $\ket{\phi_k^S}$ of Eq.~\eqref{eq:CanonicalBasis} with the substitution $\ket{E_j} \rightarrow \ket{A_j}$.
Applying $U_A$ to $\ket{A_0}$ induces the following transition
\begin{equation}
\ket{A_0} \rightarrow \ket{A_g}= \sum_{j=0}^{N_A-1}e^{i \theta_j}\sqrt{\lambda_j} \ket{A_j}.
\end{equation}
Exactly like $U_S$, even $U_A$ can be decomposed as a sequence of elementary two-dimensional rotations followed by a local phase change
\begin{equation}
U_A = \prod_{j=1}^{N_A} U_{A,j},
\end{equation}
where, for $1 \le j<N_A$,
\begin{multline}\label{eq:UAi}
U_{A,j} = \sqrt{\frac{\gamma_j}{\gamma_{j-1}}} \prt{\dyad{A_0}+\dyad{A_j}} +\sum_{k \neq 0,j} \dyad{A_k} \\
+\sqrt{\frac{\lambda_j}{\gamma_{j-1}}} \prt{e^{i \theta_j}\dyad{A_j}{A_0}-e^{- i \theta_j}\dyad{A_0}{A_j}},
\end{multline}
while, for $j=N_A$,
\begin{equation}
U_{A, N_A} = e^{i \theta_0} \dyad{A_0} + \sum_{k=1}^{N_A-1}  \dyad{A_k}.
\end{equation}
	
After having obtained $\ket{A_g}$, a single instance of the generalized CNOT gate~\cite{Wilmott2011} can be applied to obtain $\ket{\phi}$.
For the CNOT gate we use the following form~\footnote{In Ref.~\cite{Wilmott2011} the CNOT gate is given for a bipartite system composed of two qubits with the same dimensions. Here, we have used what seemed to us the most natural extension of the CNOT gate operator, i.e., using system $B$ as a qudit of dimension $N_A$.
Therefore, we apply the CNOT gate such that it involves only the first $N_A$ levels of both systems.}
\begin{equation}
\label{eq:CNOT}
U_{G_A} = \sum_{i,j=0}^{N_A-1} \dyad{A_i B_{i \oplus j}}{A_i B_j} + \sum_{i=0}^{N_A-1} \sum_{j=N_A}^{N_B-1} \dyad{A_i B_j},
\end{equation}
where $ i \oplus j = (i + j)\mod N_A$.
Therefore, the final unitary operator is given by $\tilde{U}_A = U_{G_A} U_A$.
The main advantage of using this technique is that, being $U_A$ local,
only a single standard two-qudit gate has to be used to construct $\tilde{U}_A$.
Moreover, we have provided an explicit decomposition of $U_A$ in terms of elementary operations.

A similar procedure can be implemented by applying on system $B$ a local unitary operator $U_B$ acting on system $B$ analogously to how $U_A$ acts on system $A$~\footnote{Indeed, the operator $U_B$ acts as $U_A$ on the first $N_A$ levels of system $B$, while leaving unvaried the others.} and then using a CNOT gate, $U_{G_B}$, defined analogously to $U_{G_A}$ by reversing the role of qudits $A$ and $B$, i.e.,
\begin{equation}
\label{eq:CNOTB}
U_{G_B} = \sum_{i,j=0}^{N_A-1} \dyad{A_{i \oplus j} B_{j}}{A_i B_j} + \sum_{i=0}^{N_A-1} \sum_{j=N_A}^{N_B-1} \dyad{A_i B_j}.
\end{equation}
In this case, the final unitary operator is given by $\tilde{U}_B = U_{G_B} U_B$. We remark that, contrarily to $U_S$ (or its modified form for target states orthogonal to $\ket{E_0}$, detailed in Appendix~\ref{sec:APPUnitaryOperatorConstruction}), the operators $\tilde{U}_A$ and $\tilde{U}_B$  do not allow to move from $\ket{E_0}$ to any other state .
In fact, by construction\blue{,} only states such as those of Eq.~\eqref{eq:TargetStates} can be obtained (also here the procedure should be modified in the case of target states orthogonal to $\ket{E_0}$).  

The implementation of $\tilde{U}_A$ and $\tilde{U}_B$ in quantum circuits is expected to be largely simplified by the fact that they are given in terms of simple two-level local rotations and a local change of phase on a subsystem, and, subsequently, of a generalized CNOT gate on the total system.

We finally observe that the  decompositions of the three unitary operators provided in this section are given in terms of operators that never lower the energy of system $S$.
This could be useful towards an efficient implementation of these operators since the apparatus implementing the unitary operators never has to recover energy from system $S$ during the process.
Moreover, we remark that these unitary protocols can be used even if the starting state $\rho$ is not exactly the ground state $\ket{E_0}$ but a good approximation of it, i.e., such that the fidelity $F(\rho,\dyad{E_0})$ is very near one, where the fidelity for two arbitrary states $\sigma $ and $\dyad{\varphi} $, in the case when at least one of them is pure, is equal to $F=\ev{\sigma}{\varphi}$. Indeed, one can easily check that for any suitable unitary operator $ U$ generating the requested state $U\ket{E_0}$, the final fidelity is equal to the initial one since $F(U\rho U^\dagger,U\dyad{E_0}U^\dagger)=F(\rho,\dyad{E_0})$.

\section{The zero-temperature thermalization approach\label{sec:TwoApproaches}}

In this section, we consider a different way to generate the states $\stMin$ exploiting a zero-temperature thermalization protocol.
This protocol consists in turning on a suitable interaction between systems $A$ and $B$, each of which had already thermalized to its own thermal state (ground state at $T=0$), \blue{whereas} they are also weakly coupled to a common bath at zero temperature.
After the thermalization takes place the coupling with the bath is suppressed and the interaction between $A$ and $B$ is turned off. In particular, here we make the assumption that the phases of turning on and off the interaction are so rapid that the state of the system is practically constant during the corresponding time intervals.
\blue{If one wants that after the generation protocol the state does not change in time, one can also assume that after the thermalization the coupling with the bath is suppressed.}

We considered this kind of protocols in the thermodynamic context of work extraction from a resource in Ref.~\cite{Piccione2019Work} where the efficiency of such protocols has been discussed.
At zero temperature, the ideal  efficiency and the expended energy of this protocol is given by~\cite{Piccione2019Work}
\begin{equation}
\label{eq:ProtocolEfficiency}
\eta = \frac{\ev{H_0}{\psi} - E_0}{\Eexp},
\quad
\Eexp = \ev{H_I}{E_0} - \ev{H_I}{\psi},
\end{equation}
where $H_I$ is the interaction term that we turn on and off and $\ket{\psi}$ is the ground state of the Hamiltonian $H = H_0 + H_I$.
In the above formula for $\eta$, the numerator represents the energy stored in system $S$ at the end of the protocol minus the initial one, \blue{whereas} the denominator is the minimum amount of expended energy to turn on and off the interaction.
Therefore, we can call the denominator alone $\Eexp$, standing for \enquote{expended energy}.
If our only concern is to maximize the entanglement production with respect to the used energy, then it is more important to minimize $\Eexp$ instead of maximizing $\eta$.
We stress that this quantity represent the minimum energy required to perform the thermalization protocol. It corresponds to the actual amount of expended energy only in the ideal case when all energy losses take place only during the thermalization of the system. 

To produce a MEES, $\stMin$, we need to find an interaction Hamiltonians $H_I$ making $\stMin$ the ground state of the Hamiltonian $H=H_0 + H_I$.
In the following two subsections, we deal with the problem of finding some suitable interaction Hamiltonians leading to the desired ground state through methods that can work for any bipartite system.

\subsection{\label{subsec:SimpleApproach} A simple and a modified simple approach}

The easiest total Hamiltonian leading to the desired MEES through a zero-temperature thermalization is
\begin{equation}
\label{eq:AlternativeHamiltonian}
\Hsi_S = -V_S \dyad{\psi_g}, 
\end{equation}
where $V_S>0$. This Hamiltonian has $\stMin$ as \blue{non}degenerate ground state and a degenerate excited subspace of dimension $N_S-1$.
The interaction term needed to obtain this Hamiltonian is
\begin{equation}
\label{eq:SimpleGlobalInteractionHamiltonian}
\Hsi_I = -V_S \dyad{\psi_g} - H_0.
\end{equation}

From the mathematical point of view, the parameter $V_S$ can assume any positive value. However, absolute zero temperature is never physically attainable. Therefore, $V_S$ must be high enough so that, for a sufficiently low temperature, the energy gap between the ground state and the excited states is high enough to make good the zero-temperature approximation.
In particular, with this Hamiltonian one also has to consider the high-degeneration of the excited level.
In the thermal state, the population ratio between the ground state and the degenerate excited level is
\begin{equation}\label{eq:pe/pg}
\frac{p_e}{p_g} = (N_S-1) e^{-\beta V_S},
\end{equation}
where $\beta=1/\prt{k_B T}$, $k_B$ is the Boltzmann constant, and $T$ is the temperature of the bath.
Therefore, the larger the subsystems $A$ and $B$ are the higher $V_S$ has to be in order to make the above ratio sufficiently small, consistently with the zero-temperature approximation.

The Hamiltonian of Eq.~\eqref{eq:AlternativeHamiltonian} can be improved by considering that $\stMin$ has not all the components in the bare basis, but only $N_A$ components.
For simplicity, let us consider that $A_0=B_0=0$ and $A_1,B_1 > 0$.
Then, one could use the following interaction term:
\begin{equation}
\label{eq:SimpleInteractionHamiltonianModified}
\Hsim = -V_M \dyad{\psi_g} - \sum_{i=0}^{N_A-1} E_i \dyad{E_i},
\end{equation}
where $V_M>0$.
By using this interaction term, the total Hamiltonian $H^{\textup{si}}_M=H_0+\tilde{H}^{\textup{si}}_I $ has again $\stMin$ as its ground state but the degeneracy of the first excited level is only $N_A-1$.

Let us define $\Delta = \min(A_1,B_1)$.
Now, we make the assumption that the first excited level of $H_0$ is \blue{non}degenerate and that the other excited levels are enough higher than the first one so that the condition $e^{-\beta \Delta} \ll 1$ implies that the zero-temperature approximation for $H_0$ is valid.
The value of $V_X$ (where $X=S, M$ depending on which case is analyzed) such that the population ratio between the ground and first excited level is the same for $H_0$ and for $H^{\textup{si}}_X$ [cfr. Eq.~\eqref{eq:pe/pg}] results equal to
\begin{equation}
\label{eq:EstimationOfV}
V_X \simeq \Delta \prtq{1 + \frac{\ln(N_X - 1)}{\beta \Delta}},
\end{equation}
where $N_S$ had been already defined as $N_S=N_A N_B$, and $N_M=N_A$. When $V_X$ has such value, we can be certain that the zero-temperature approximation holds also after having turned on the interaction term.

When using this interaction term to enable the protocol (again with $A_0=B_0=0$ for simplicity), its efficiency and the expended energy for the generation of a MEES $\ket{\psi_g}$ are easily obtained from Eq.~\eqref{eq:ProtocolEfficiency} as
\begin{equation}
\label{eq:ProtocolEfficiency2}
\eta = \frac{\Emin}{\Eexp},\quad
\Eexp = \Emin + V_X \prt{1-\lambda_0},
\end{equation}
where we recall that the energy $\Emin$ of the MEESs is given in Eq.~\eqref{eq:minimumenergy} and where $\lambda_0$ can be computed  from Eq.~\eqref{eq:MinimumState}. For  the generation of states $\ket{\phi}$ having the form of Eq.~\eqref{eq:TargetStates}, Eq.~\eqref{eq:ProtocolEfficiency2} holds by replacing $\Emin$ with the energy of the state $\ket{\phi}$ and using for $\lambda_0$ its actual value for this state.

We notice that, using Eq.~\eqref{eq:SimpleGlobalInteractionHamiltonian}, any arbitrary state of system $S$ can be generated as the result of the thermalization by replacing $\stMin$ with the desired state $\ket{\psi}$.
Therefore, it is possible to compare how MEESs behave with respect to all the other states.
We already know from Ref.~\cite{Piccione2020Energy} that, for a fixed amount of entanglement, $\ev{H_0}{\psi}$ is minimized by $\stMin$.
On the other hand, for a fixed entanglement, $V_X\prt{1-\lambda_0}$ is minimized by taking the largest possible value of $\lambda_0$.
Therefore, from a mathematical point of view, the states which minimize $\Eexp$ for a given entanglement are not the MEESs.
However, in all the simulations we have run, as in the example shown in Sec.~\ref{sec:ApplTherm}, the $\Eexp$ associated to the MEESs result to be very close to the real minima of $\Eexp$ so that in practical applications one can directly choose to generate the MEESs if the aim is to minimize the expended energy.
Specializing to the case of Eq.~\eqref{eq:SimpleInteractionHamiltonianModified}, the analytical solution for the states minimizing $\Eexp$ can be found, as shown in Appendix~\ref{sec:APPMinimizationSimple}. We also remark that using this simple modified approach, only states like those of Eq.~\eqref{eq:TargetStates} can be obtained by replacing $\ket{\psi_g}$ with $\ket{\phi}$ in Eq.~\eqref{eq:SimpleInteractionHamiltonianModified}. 

\subsection{\label{subsec:UnitaryApproach}A unitary transformation approach}

Another way to obtain a new Hamiltonian having $\stMin$ as its ground state is to apply a transformation of the kind $U H_0 U^\dagger$ to the original Hamiltonian $H_0$, with a unitary operator $U$ such that $U \ket{E_0} = \stMin$.
This has the advantage that when both systems $A$ and $B$ can be considered in their ground state, i.e., the zero-temperature approximation is well satisfied for both of them, then it will be valid also after the turning on of the interaction term since a unitary transformation does not change the spectrum of the Hamiltonian $H_0$. It follows that the zero-temperature thermalization will bring the system to the new ground state $\stMin$.
Within this approach, the interaction term assumes the form:
\begin{equation}
\label{eq:InteractionHamiltonian}
H_I = U H_0 U^\dagger - H_0.
\end{equation}

We already know three different unitary operators assuring the requested transition,  $U \ket{E_0} = \stMin$, from Sec.~\ref{sec:UnitaryTransformation}.
In the next two subsections we describe the interaction Hamiltonians that they generate.

\subsubsection{$H_I$ from the \nonloc{} unitary operator}

Applying the unitary operator $U_S$ of Eq.~\eqref{eq:UnitaryOperatorGlobal} in Eq.~\eqref{eq:InteractionHamiltonian}, one can obtain the explicit form of $H_I$. This is done in Appendix~\ref{sec:APPExplicitFormHI} for a general target state $\ket{\psi}$. Here, we report the results in the case of a target state of the form of Eq.~\eqref{eq:TargetStates} by only focusing on the $N_A \times N_A$ terms that are non trivial, i.e., all the terms of the kind $\dyad{A_i B_i}{A_j B_j}$.
The remaining part of the matrix is just filled with zeroes.
Denoting $X_0 = \ev{H_I}{E_0}/\lambda_0$, where
\begin{equation}
\ev{H_I}{E_0} = E_0 (\lambda_0 - 1) + \lambda_0 \sum_{i=1}^{N_A-1} \frac{\lambda_i E_i}{\gamma_i \gamma_{i-1}},
\end{equation}
\begin{equation}
X_i \equiv E_0 - \frac{E_i}{\gamma_{i-1}} + \sum_{k=1}^{i-1} \frac{\lambda_k E_k}{\gamma_k \gamma_{k-1}},
\end{equation}
\blue{for $1 \leq i \leq N_A-1$}, and $\Lambda_{i,j} = e^{i(\theta_j - \theta_k)}\sqrt{\lambda_j \lambda_k}$, we arrive at the matrix of $H_I$,
\begin{equation}\label{eq:InteractionMatrix}
H_I = \mqty(
\Lambda_{0,0} X_0 & \Lambda_{0,1} X_1 & \Lambda_{0,2} X_2 &\dots & \Lambda_{0,n} X_{n} \\
%%%%%
\Lambda_{1,0} X_1 & \Lambda_{1,1} X_1 & \Lambda_{1,2} X_1 & \dots & \Lambda_{1,n} X_{1} \\
%%%%%
\Lambda_{2,0} X_2 & \Lambda_{2,1} X_1 & \Lambda_{2,2} X_2 & \dots & \Lambda_{2,n} X_{2} \\
%%%%%
\vdots & \vdots & \vdots & \ddots & \vdots \\
%%%%%
\Lambda_{n,0} X_{n} & \Lambda_{n,1} X_{1} & \Lambda_{n,2} X_2 &\dots & \Lambda_{n,n} X_{n}
),
\end{equation}
where $n=N_A-1$.

Using Eqs.~\eqref{eq:CanonicalBasis}, \eqref{eq:UnitaryOperatorGlobal}, \eqref{eq:ProtocolEfficiency}, and \eqref{eq:InteractionHamiltonian}, and setting $E_0=0$, the efficiency of the thermalization protocol for target states like those of Eq.~\eqref{eq:TargetStates} results to be equal to (for more details, see Appendix~\ref{sec:APPEfficiencyUnitaryGlobal}, where the computation is performed for an arbitrary state $\ket{\psi}$) 
\begin{equation}
\label{eq:EfficiencyGlobalUnitary}
\eta = \frac{\sum_{i=1}^{N_A-1} \lambda_i E_i }{\Eexp^S},
\
\Eexp^S =\sum_{i=1}^{N_A-1} \lambda_i E_i \prtq{1+\frac{\lambda_0}{\gamma_i \gamma_{i-1}}}.
\end{equation}
One can also show that the expended energy $\Eexp^S$ cannot be higher than $2 E_{N_A-1}$  (see Appendix~\ref{sec:APPEfficiencyUnitaryGlobal} for the derivation of this inequality in a more general case).

\subsubsection{$H_I$ from the \loc{} unitary operators \label{subsec:HI from the local unitary operators}}

Other two interaction terms can be obtained by the application of the local unitaries $U_A$ and $U_B$ of Sec.~\ref{subsec:LocalUnitaryApproach}, together with the, respective, generalized CNOT gate.
Since the two \loc{} operations have the same structure, we will concentrate on the case of $\tilde{U}_A$.

The operator $U_A$ acts on $\ket{A_i}$ as $U_S$ acts on $\ket{E_i}$, therefore one can write the matrix $U_A H_A U_A^\dagger + H_B$ as [see Eqs.~\eqref{eq:InteractionHamiltonian} and \eqref{eq:InteractionMatrix}] 
\begin{equation}
\sum_{i,j=0}^{N_A-1} \sum_{k=0}^{N_B-1} \prtq{\delta_{i,j} \prt{A_i+B_k} + \Lambda_{i,j}X^A_{i,j}}\dyad{A_i B_k}{A_j B_k},
\end{equation}
where $X^A_{i,j}$ refers to the correspondent $X_\alpha$ of the matrix of Eq.~\eqref{eq:InteractionMatrix} with all the $E_i$ replaced by $A_i$ [the coefficients $\lambda_i$ remain the same, even for the case of Eq.~\eqref{eq:MinimumState}], $\delta_{i,j}$ is the Kronecker \blue{$\delta$}, and
\begin{equation}
\alpha =
\begin{cases}
\max(i,j), &\qq{if} ij=0,\\
\min(i,j), &\qq{if} ij\neq 0.
\end{cases}
\end{equation}
Then, applying the generalized CNOT gate $U_{G_A} $ and subtracting the original Hamiltonian $H_0$, we obtain
\begin{multline}
H_I= \sum_{i,j,k=0}^{N_A-1} \prtbg{\\
\prtq{\delta_{i,j} \prt{B_k-B_{i \oplus k}} + \Lambda_{i,j}X^A_{i,j}} \dyad{A_i B_{i \oplus k}}{A_j B_{j \oplus k}}}\\
+ \sum_{i,j=0}^{N_A-1} \sum_{k=N_A}^{N_B-1} \Lambda_{i,j}X^A_{i,j} \dyad{A_i B_{k}}{A_j B_{k}}.
\end{multline}

Even for this approach, we can calculate the efficiency of the thermalization protocol and the expended energy for target states like those of Eq.~\eqref{eq:TargetStates} (see Appendix~\ref{sec:APPEfficiencyLocalUnitary} for the calculations).
Setting $A_0=B_0=0$, we get
\begin{equation}
\label{eq:EnergyExpenseA}
\eta = \frac{\sum_{i=1}^{N_A-1} \lambda_i E_i }{\Eexp^A},
\
\Eexp^A = \sum_{i=1}^{N_A-1} \prtq{\lambda_i \prt{ E_i + \frac{\lambda_0 A_i}{\gamma_i \gamma_{i-1}}}}.
\end{equation}

Notice that the \loc{} approach based on $\tilde{U}_A$ is always more efficient than the \nonloc{} one when aiming at the same state, since
\begin{equation}
\label{eq:DifferenceLocalGlobal}
\Eexp^S - \Eexp^A = \lambda_0 \sum_{i=1}^{N_A-1} \frac{\lambda_i B_i}{\gamma_i \gamma_{i-1}} \ge 0,
\end{equation}
while the numerators of the efficiencies are the same in the two cases. This conclusion also holds for the \loc{} approach based on $\tilde{U}_B$, for which a formula analogous to the one of Eq.~\eqref{eq:EnergyExpenseA} applies where the quantities $A_i$ and  $B_i$ are swapped.

With regard to which approach is more efficient between the two possible \loc{} approaches, the answer is model dependent as it strongly depends on the spectra of both local Hamiltonians.
In this case, the difference is given by 
\begin{equation}
\label{eq:ExpenseLocalDifference}
\Eexp^B - \Eexp^A = \lambda_0 \sum_{i=1}^{N_A-1} \frac{\lambda_i \prt{B_i-A_i} }{\gamma_i \gamma_{i-1}}.
\end{equation}

As we will see in the example of Sec.~\ref{sec:ApplTherm}, which of the two approaches requires more energy depends also on how much entanglement is demanded.

\section{\label{sec:ApplTherm} Analysis of the thermalization protocols for a three by four system}

The simpler system to which we could apply the thermalization protocol is a system composed of two qubits.
However, the two-level structure makes most results trivial and not meaningful in general so that the analysis of this section is devoted to a bigger system case. For completeness, in Appendix~\ref{sec:TwoQubits} we report all the explicit formulas for a two-qubit system.

In Ref.~\cite{Piccione2020Energy} we analyzed a $3\times 4$ system with spectra $\sigma(H_A)=\prtg{0,2,4}$ and $\sigma(H_B)=\prtg{0,1,6,9}$ in arbitrary units.
In this section, we analyze the thermalization protocol applied to the same system, focusing on the behaviour of the MEESs with respect to all the other states that every specific interaction Hamiltonian we devised could produce.

\begin{figure}
	\centering
	\includegraphics[width=0.48\textwidth]{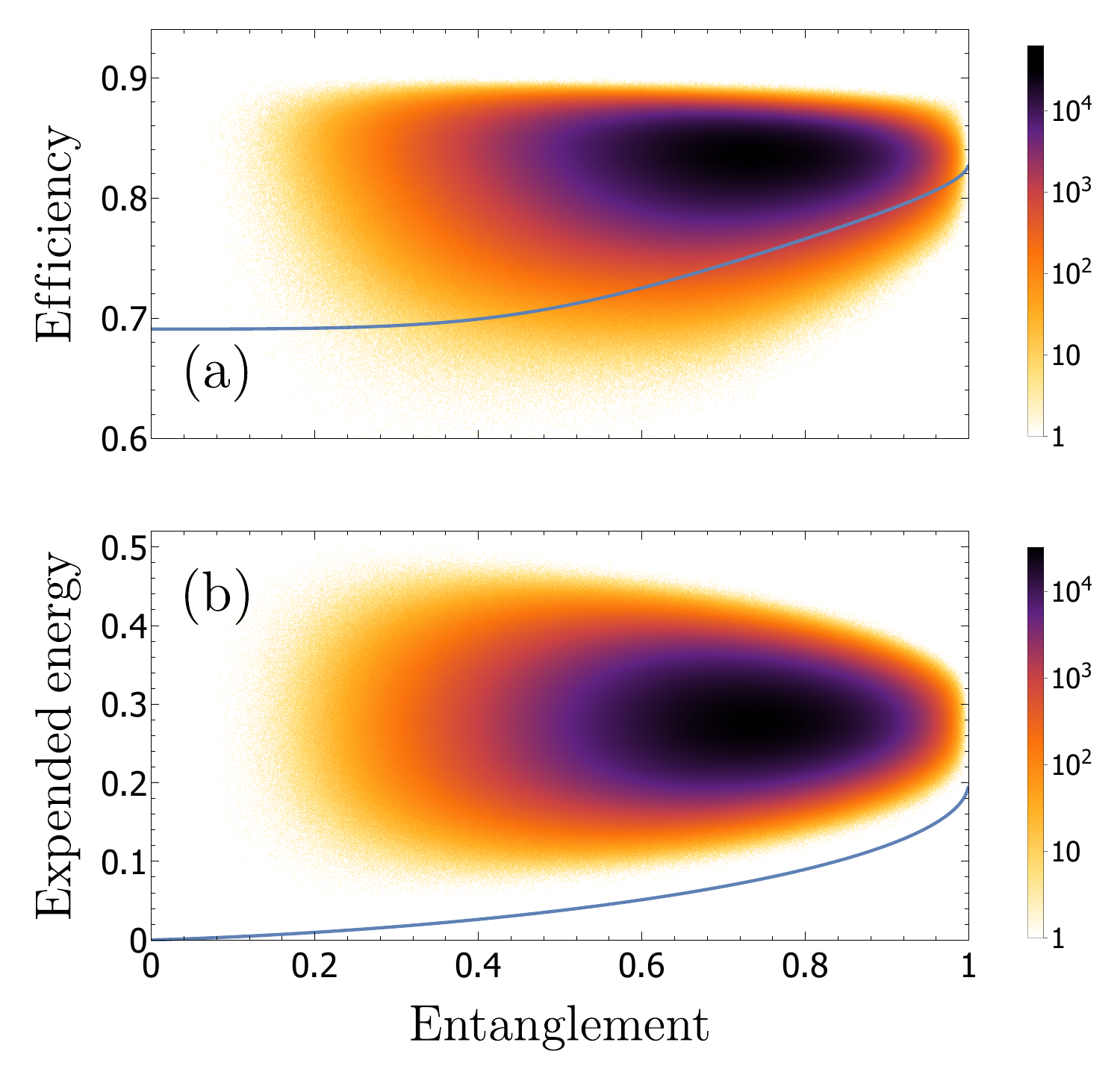}
	\caption{
		Distribution of randomly generated pure states with respect to the entropy of entanglement, $\E$, and the efficiency of the protocol (a) or the amount of energy used to create them (b).
		For these plots, the thermalization protocol is based on the simple approach.
		The relevant Hamiltonians have spectra $\sigma(H_A)=\prtg{0,2,4}$ and $\sigma(H_B)=\prtg{0,1,6,9}$ in arbitrary units.
		Both the entanglement and the expended energy are normalized to $1$, respectively, with respect to $\ln(3)$ and $2\max{\sigma(H_0)}$.
		Each graph is the result of an interpolation of $10^9$ random states distributed in a $1000 \times 1000$ grid, which covers the whole range of values (even those not shown on the graph; e.g., the grid of efficiency goes from zero to one). The colors correspond to $\log_{10}(1+c)$, where $c$ is the number of states in each grid  element. In the bar legend, we report the value of $1+c$, so that, e.g., 1 corresponds to the case of counting equal to zero.	
		The blue lines give the position of the MEESs.
		Regarding the efficiency, we notice that it is quite high for every state and that the MEESs are among the worst ones.
		However, regarding the expended energy, they are much cheaper than the vast majority of all the possible pure states that can be generated.}
	\label{fig:ComparisonSimpleGlobal}
\end{figure}

Let us start from analyzing the simple case interaction of Eq.~\eqref{eq:SimpleGlobalInteractionHamiltonian}.
Since in Eq.~\eqref{eq:AlternativeHamiltonian} $\stMin$ can be replaced by any state of system $S$, any state can be obtained through a zero-temperature thermalization associated to $\Hsi_S$.
Figure~\ref{fig:ComparisonSimpleGlobal} shows two plots relating the entanglement of $10^9$ randomly generated states to, respectively, their generation efficiency $\eta$ and their energy cost $\Eexp$.
The MEESs, represented by the blue lines, are not the best ones with respect to the efficiency, but, in general, they cost much less than the vast majority of the all the possible pure states that can be generated.
In effect, even if they are not the states mathematically minimizing the expended energy, numerical solutions of the minimization problem show, in all systems that we have simulated, that they are very close to them.

\begin{figure}
	\centering
	\includegraphics[width=0.48\textwidth]{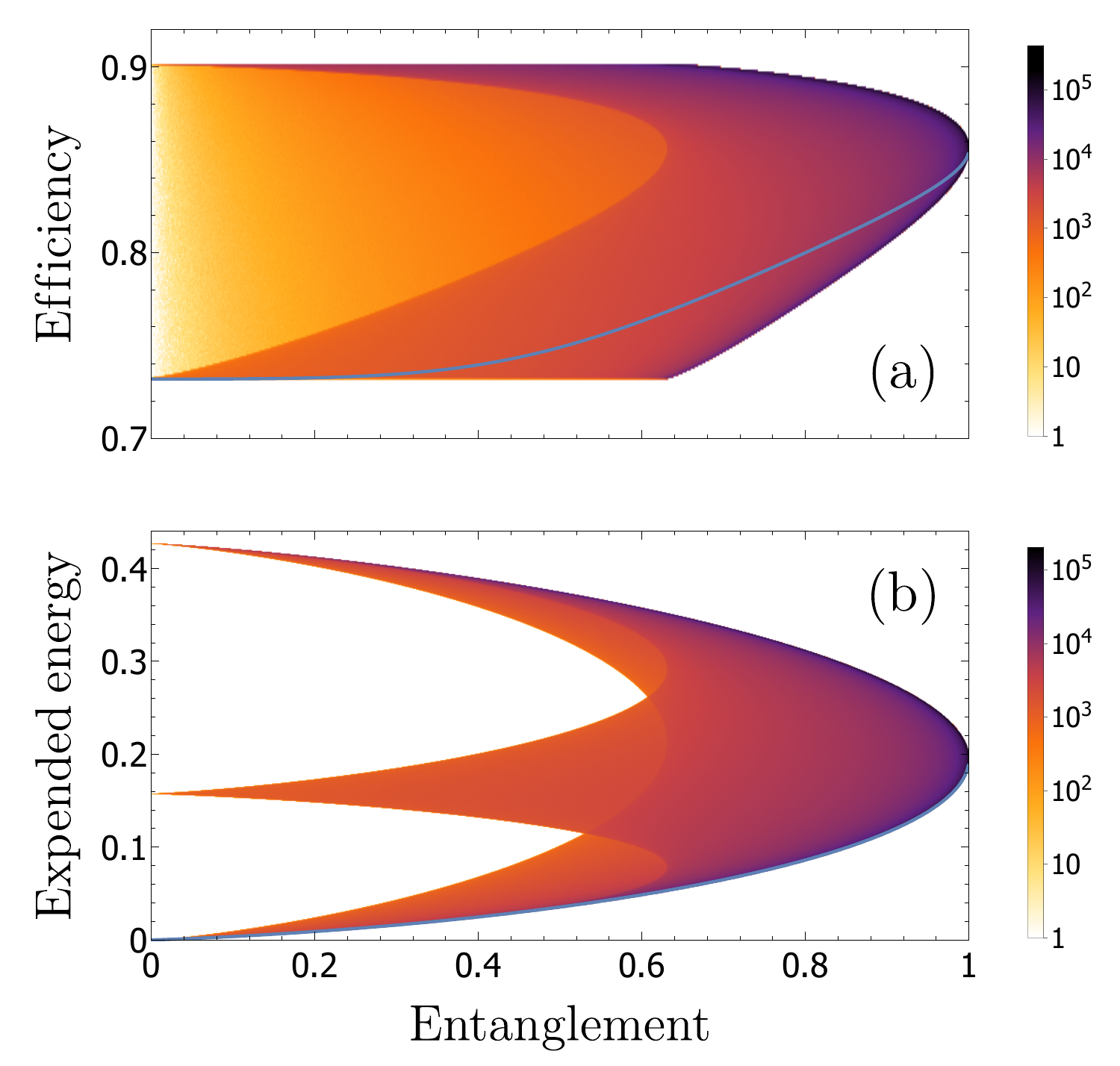}
	\caption{\blue{The} same plots of Fig.~\ref{fig:ComparisonSimpleGlobal} but with the thermalization protocol based on the modified simple approach.
	Here, we have used a different set of states with respect to the case of Fig.~\ref{fig:ComparisonSimpleGlobal}, since the modified simple approach cannot generate all the states but only those having the form of Eq.~\eqref{eq:TargetStates}. In particular, we have randomly generated $10^9$ states of this latter kind.
	As in Fig.~\ref{fig:ComparisonSimpleGlobal}, the efficiency is quite high for every state, with the MEESs being among the worst ones but also costing less than the vast majority of all the possible pure states that can be obtained.}
	\label{fig:ComparisonSimpleLocal}
\end{figure}

With the interaction Hamiltonian of the modified simple approach of Eq.~\eqref{eq:SimpleInteractionHamiltonianModified}, since $V_M < V_S$ [see Eq.~\eqref{eq:EstimationOfV}], we get an higher efficiency compared to the one we could get by reaching the same states with the \blue{nonmodified} simple approach. 
However, the modified simple approach lets us to obtain only states like those of Eq.~\eqref{eq:TargetStates}.
Figure~\ref{fig:ComparisonSimpleLocal} shows the same two plots of Fig.~\ref{fig:ComparisonSimpleGlobal} but using a sample of $10^9$ randomly generated states that can be obtained through the modified simple approach.
Indeed, the difference of the sample set is the main responsible for the great difference between the distributions of Fig.~\ref{fig:ComparisonSimpleGlobal} and Fig.~\ref{fig:ComparisonSimpleLocal}.
Even in this case, the MEESs, represented by the blue lines, are not the best ones with respect to the efficiency.
Regarding the expended energy, they are still among the cheapest states, but the distribution accumulates on the boundaries contrarily to the case of Fig.~\ref{fig:ComparisonSimpleGlobal}(b).
By zooming enough in the graph, it is possible to see (using a thinner line for the MEESs) that some of the generated states are cheaper than the MEESs.
We recall that, in this case only, we were able to find a simple analytical solution of the minimization problem, reported in Appendix~\ref{sec:APPMinimizationSimple}.
As in the simple approach case of Fig.~\ref{fig:ComparisonSimpleGlobal}, the solutions of the minimization problem show, in all systems that we have simulated, that MEESs are very close to the states minimizing the expended energy.
We also notice that the maximum value of the efficiency with respect to the normalized entanglement starts to change around $\ln(2)/\ln(3) \simeq 0.631$.
This is probably due to the fact that to obtain a certain degree of entanglement without involving all three states ($\ket{E_0},\ \ket{E_1}$, and $\ket{E_2}$) is not possible when $\E > \ln(2)$.

\begin{figure}[t!]
	\centering
	\includegraphics[width=0.48\textwidth]{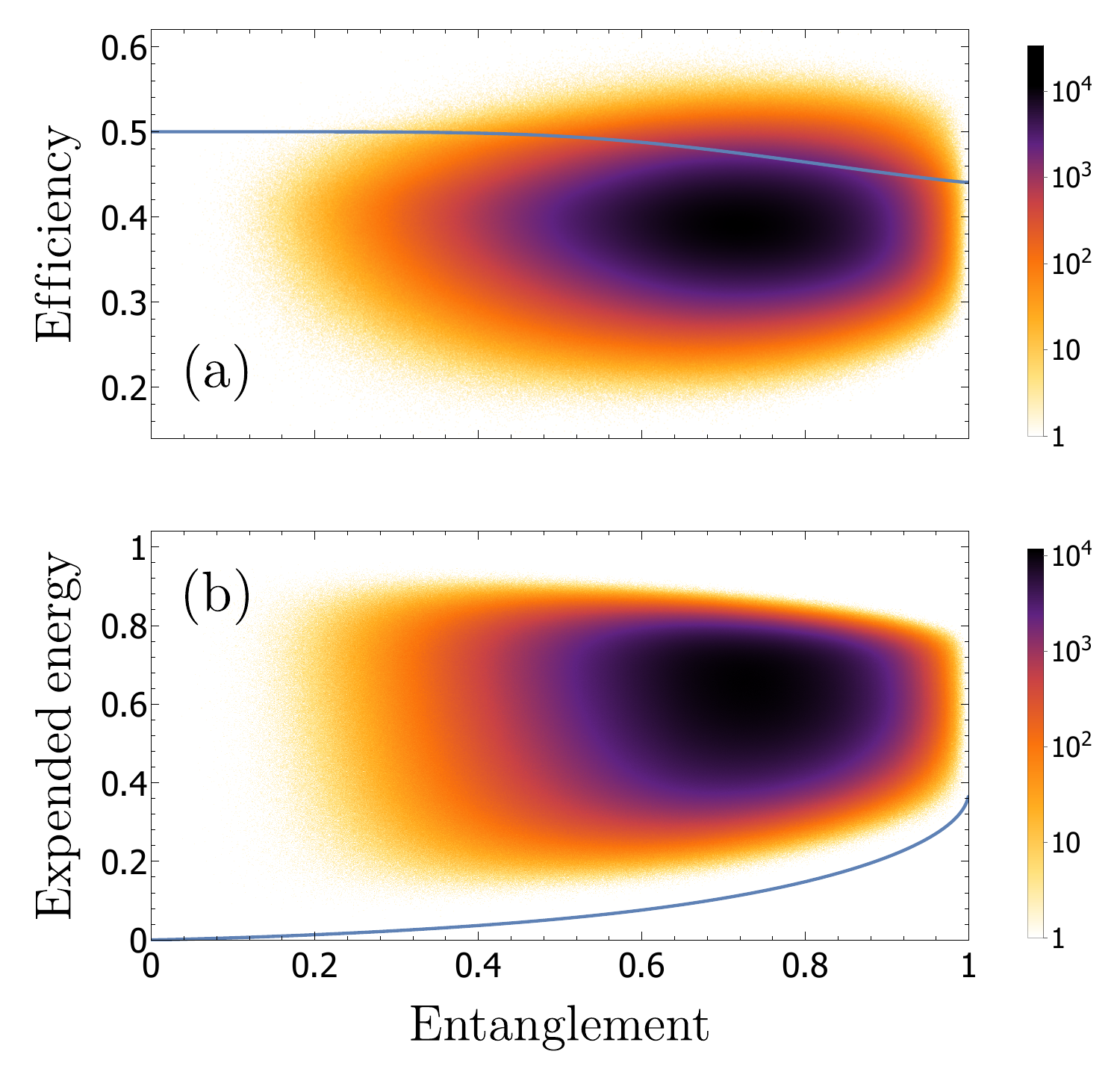}
	\caption{\blue{The same} plots of Fig.~\ref{fig:ComparisonSimpleGlobal} but with the thermalization protocol based on the transformation of $H_0$ through $U_S$.
	Here, the same states of Fig.~\ref{fig:ComparisonSimpleGlobal} have been used.
	Compared to Fig~\ref{fig:ComparisonSimpleGlobal}, the efficiency is remarkably lower, but in this case the MEESs are among the best ones.
	As in all the other figures, they are cheaper than the vast majority of all the possible pure states that can be generated.}
	\label{fig:ComparisonUnitaryGlobal}
\end{figure}

Moving to the results of the unitary approach, Fig.~\ref{fig:ComparisonUnitaryGlobal} shows the same two plots of Fig.~\ref{fig:ComparisonSimpleGlobal} but in the case when the interaction Hamiltonian is obtained through $U_S$.
Notice how the state distribution of Fig.~\ref{fig:ComparisonUnitaryGlobal} is qualitatively similar to that of Fig.~\ref{fig:ComparisonSimpleGlobal}. This is due to the fact that the plots in the two figures have been made based on the same kind of sample set.
In general, a comparison with all the other figures of this section suggests that the \nonloc{} unitary approach is the less performing.

Figure~\ref{fig:ComparisonLocalA} shows the same plots of previous figures in the case when the interaction is obtained through $\tilde{U}_A$.
As for the modified simple approach, only states with a Schmidt decomposition in the $\ket{E_i}$ basis can be obtained through this approach [see Eq.~\eqref{eq:TargetStates}].
Indeed, again for this reason, the distribution is completely different with respect to approaches that can give any state as result.
In general, the efficiencies are higher compared to the ones obtained with the \nonloc{} approach.
This is probably connected to the fact, proved analytically, that when the \loc{} and \nonloc{} unitary approaches generate the same state, the \loc{} ones are more efficient [see Eq.~\eqref{eq:DifferenceLocalGlobal}]. However, even in this case, the MEESs are not among the best ones but they lie very near the lower border of the $\Eexp-\E$ distribution.

\begin{figure}[t!]
	\centering
	\includegraphics[width=0.48\textwidth]{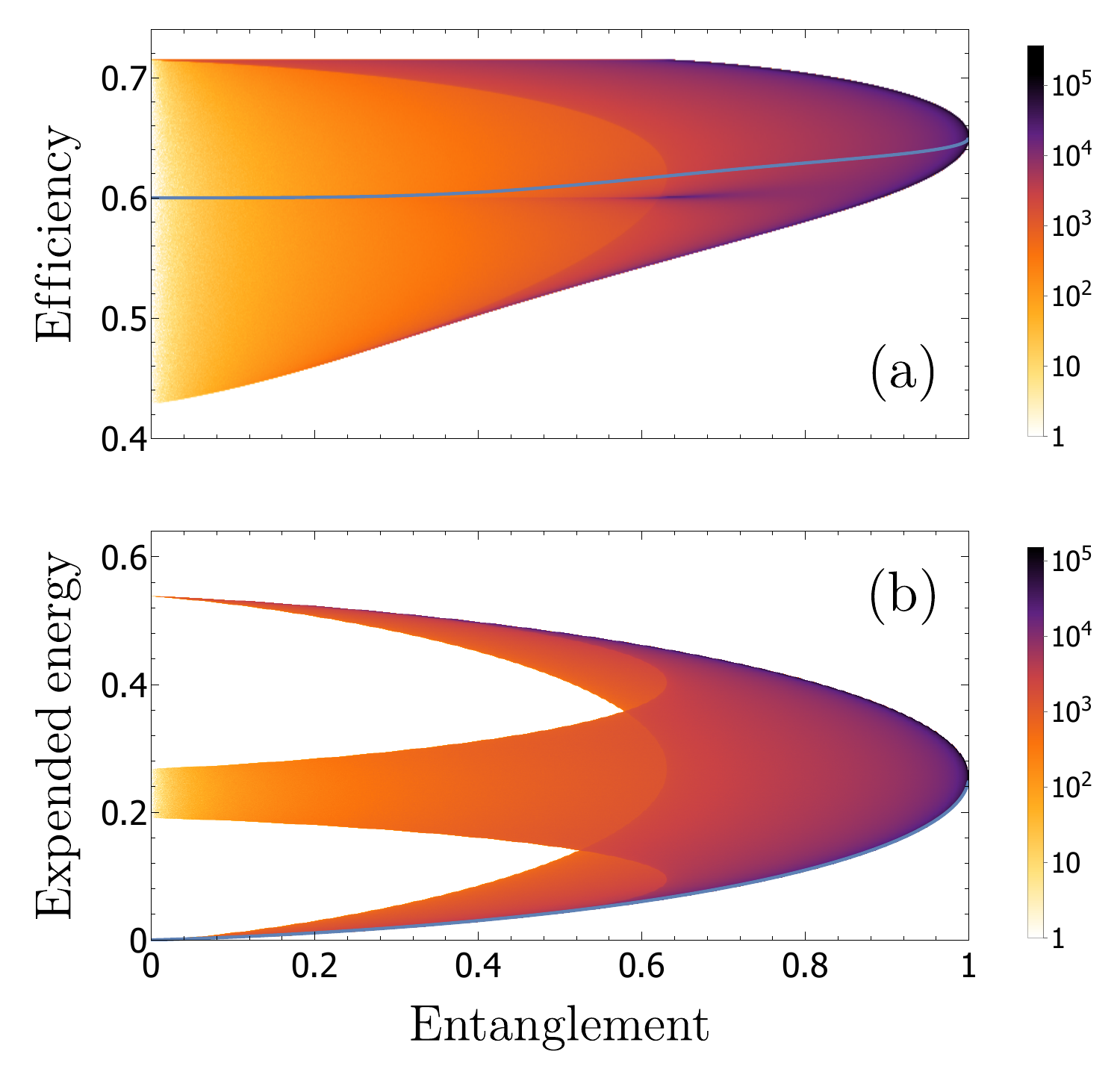}
	\caption{\blue{The same} plots of Fig.~\ref{fig:ComparisonSimpleGlobal} but with the thermalization protocol based on the transformation of $H_0$ through $\tilde{U}_A$. Here, the same states of Fig.~\ref{fig:ComparisonSimpleLocal} have been used. With respect to the unitary \nonloc{} approach of Fig.~\ref{fig:ComparisonUnitaryGlobal}, the efficiency is generally higher. The MEESs result to be cheaper to generate  than the vast majority of all the possible pure states that can be obtained. The striking difference in the distribution comes from the difference of the sample set of states.}
	\label{fig:ComparisonLocalA}
\end{figure}

\begin{figure}[h]
	\centering
	\includegraphics[width=0.48\textwidth]{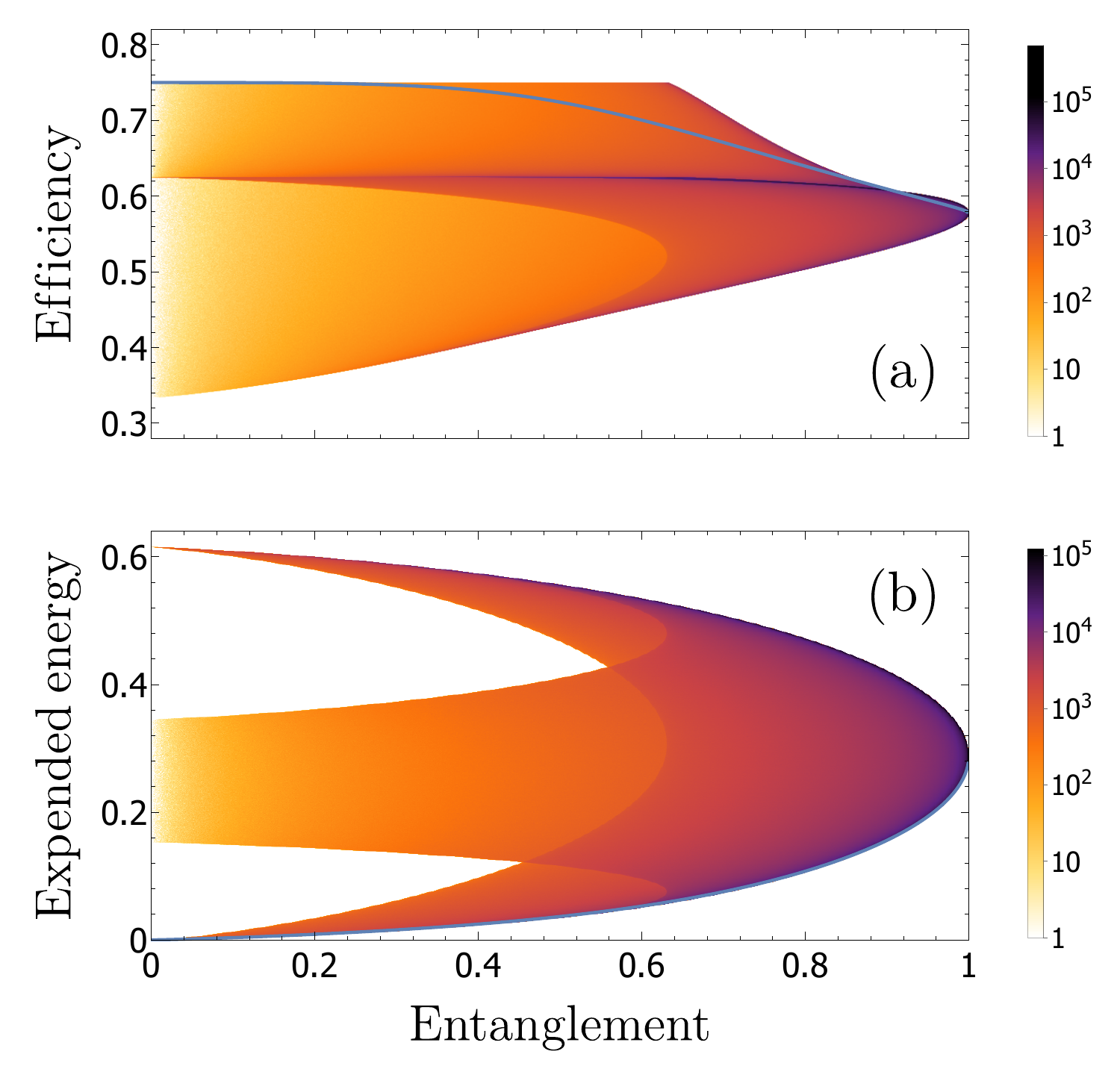}
	\caption{\blue{The same} plots of Fig.~\ref{fig:ComparisonSimpleGlobal} but with the thermalization protocol based on the transformation of $H_0$ through $\tilde{U}_B$. Here, the same states of Fig.~\ref{fig:ComparisonSimpleLocal} have been used.	
		This is the only case in which MEESs practically attain the maximum efficiency, even if only for the low entanglement zone. Concerning the expended energy, the MEESs are cheaper to generate than the vast majority of the all the possible pure states that can be generated.}
	\label{fig:ComparisonLocalB}
\end{figure}

The same plots for the unitary approach on system $B$, reported in Fig.~\ref{fig:ComparisonLocalB}, give different results for the efficiency but similar for the expended energy. We can deal with this case by using directly the operator $\tilde{U}_B$ derived in Sec.~\ref{subsec:LocalUnitaryApproach}, even if we did not report explicitly the interaction Hamiltonian needed for implementing the thermalization protocol in Sec.~\ref{subsec:HI from the local unitary operators}.
In this case, for low entanglement, the efficiencies of the MEESs are practically the highest ones. 
Other simulations suggest that this probably depends on the fact that between the two subsystems $A$ and $B$, system $B$ is the one with the lowest energy gap between the ground and the first excited state and the highest energy gap between the first and the second excited states.
In fact, within the \loc{} unitary approaches, the energy spectrum of the local system of interest is more important than the energy spectrum of the other one [see Eq.~\eqref{eq:EnergyExpenseA}].
Since for low entanglement, two states in the decomposition are sufficient, we can (heuristically) expect that the weight of the third state has to be marginal.
Indeed, considering the spectrum of system $B$ and the weights implied by a thermal distribution, the second excited state in MEESs is very low populated for low entanglement, thus leading to high efficiencies.
As for the modified simple approach, we can also see in Fig.~\ref{fig:ComparisonLocalA} and \ref{fig:ComparisonLocalB} that the maximum of the efficiency with respect to the entanglement starts declining roughly at $\ln(2)/\ln(3)$, we think, for the same reasons given in the case of Fig.~\ref{fig:ComparisonSimpleLocal}.

Finally, Fig.~\ref{fig:ComparisonMinima} shows both efficiency and expended energy for all of the analyzed processes but focusing only on the generation of MEESs.
Here, we can see that at low entanglement values the most efficient approach to produce a certain MEES is the \loc{} unitary approach \blue{based on $\tilde{U}_B$.}
For higher entanglement values, the best approach seems to be given by the modified simple one.
The plots also make clear that the two simple approaches behave the same for MEESs, with the difference originating only from $V_S > V_M$.
We also notice that the two \loc{} approaches cross at a certain point.
This is probably due to the fact the, towards higher entanglement values, the second excited level of both subsystems becomes more important.
Therefore, the only positive term in the sum of Eq.~\eqref{eq:ExpenseLocalDifference}, i.e., the last one becomes higher.
As predicted analytically, the \nonloc{} unitary approach always performs worse than the \loc{} ones.

\begin{figure}[t]
	\centering
	\includegraphics[width=0.48\textwidth]{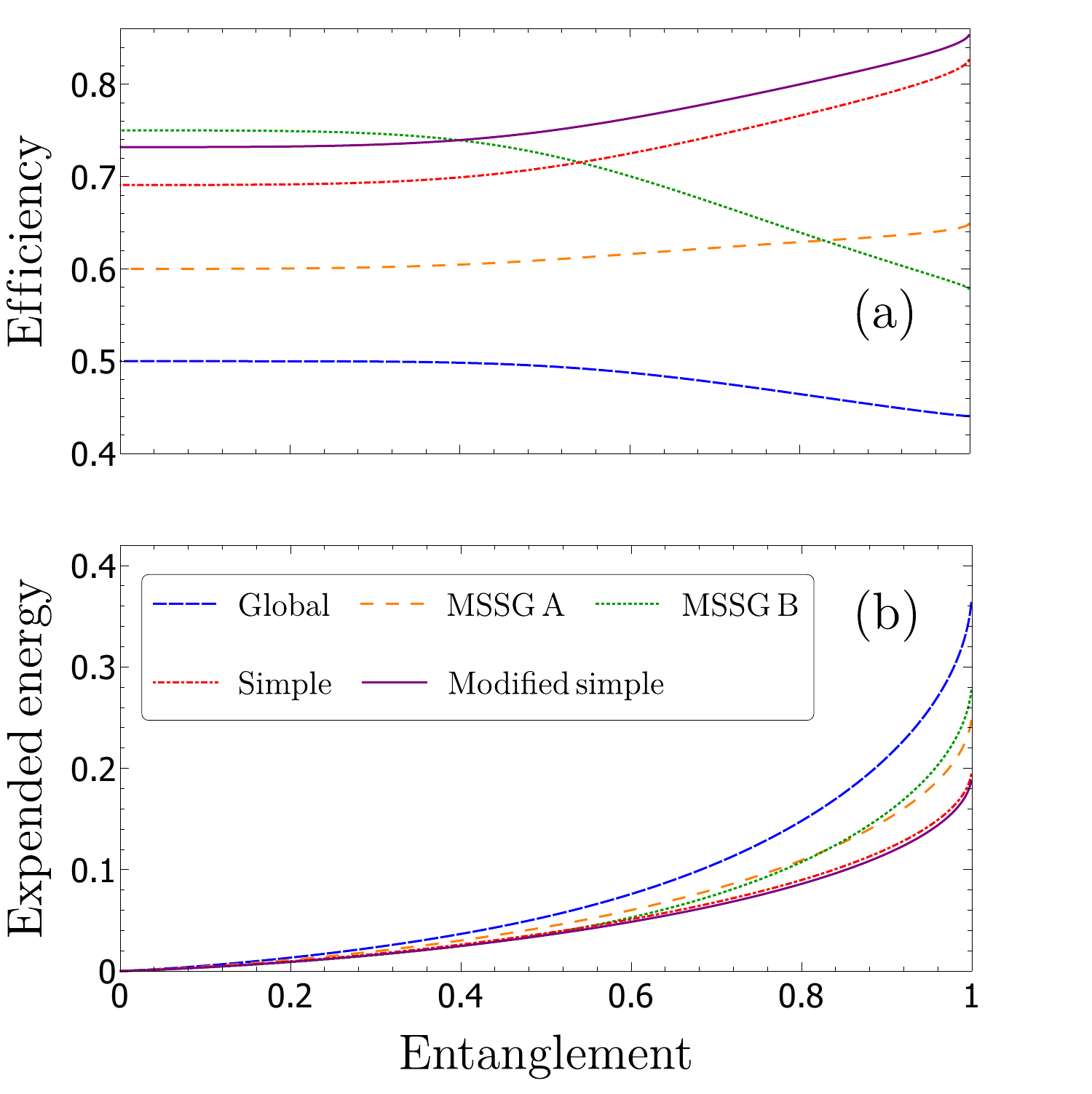}
	\caption{
		Comparison of the five approaches analyzed in Sec.~\ref{sec:ApplTherm}.
		For both plots, entanglement, quantified by the entropy of entanglement normalized to 1, is reported on the $x$-axis, while on the $y$-axis are shown, respectively, efficiency (a) or expended energy (b).
		Depending on how much entanglement is required, the most efficient states to generate can be obtained through the \blue{MSSG unitary approach based on $\tilde{U}_B$} or the modified simple one.
		We also observe that, as predicted analytically, the modified simple approach always performs better than the \blue{nonmodified} one and that the unitary approach based on $U_S$ is worse than the other two for any value of the entanglement.
	}
	\label{fig:ComparisonMinima}
\end{figure}

\section{\label{sec:Conclusions}Conclusions}

In this paper we have proposed several protocols to generate MEESs, i.e., states having the minimum energy for any given amount of entanglement, in a bipartite system $S$ made by two \blue{non}interacting parts of arbitrary finite dimensions. Some of these protocols are based on the direct use of unitary operators, while the others exploit a zero-temperature thermalization for their realization. 
	
Firstly, we have provided three different unitary operators connecting the ground state of $S$ to an arbitrary MEES. While the first operator, $U_S$, can be decomposed as the product of \blue{non}local operators, the others, $\tilde{U}_A$ and $\tilde{U}_B$, can be decomposed as the product of simple two-level local rotations and a local change of phase on a subsystem, and, subsequently, a generalized CNOT gate on the bipartite system.
Therefore, the implementation of these unitary transformations in quantum circuits should be easier with respect to the case of $U_S$.

Secondly, we have identified five different interaction Hamiltonians that, added to the free ones, make possible to generate MEESs by means of a zero-temperature thermalization process. Two of these processes are based on what we have called a simple and a modified simple approach, \blue{whereas} the others are based on the unitary operators previously identified,  $U_S$, $\tilde{U}_A$, and $\tilde{U}_B$.

We have then compared the efficiency of these zero-temperature generation processes as well as the expended energy to run them, both in general and in the case of  a specific $3 \times 4$ system, complex enough to unveil the difference between the various protocols. In doing so, we have exploited the relevant fact that our protocols can be also applied to generate states different from the MEESs. In particular, some methods can be used to reach any state of system $S$ (with proper modifications when needed), \blue{whereas} others can only give states the Schmidt decomposition of which is in the basis $\ket{E_i}$ [cfr. Eq.~\eqref{eq:TargetStates}]. By means of a detailed comparison, we have identified, in the specific $3 \times 4$ system, which are the best-performing protocols. Depending on how much entanglement is required and the spectra of the systems, the better approaches are, in general, one between those based on $\tilde{U}_A$ and $\tilde{U}_B$, or the modified simple one. Even if the MEESs are not the states minimizing the expended energy to run the generation protocol, we have numerically found, in all systems that we have simulated, that they are very close to them, so that, they are, in general, cheaper to generate than the vast majority of other states. Concerning an experimental implementation of the corresponding interaction Hamiltonians, one can easily find the matrix for all five approaches and realize which is the easiest one to implement.

\appendix

\section{\label{sec:APPUnitaryOperatorConstruction}Construction of the unitary operator in the general case}

\begin{theorem}
\label{th:OrthogonalBasis}
Be $\prtg{\ket{e_j}}_{j=0}^{N-1}$ a basis for an $N-$dimensional Hilbert space and be $\ket{\psi_0} = \sum_{j=0}^{N-1} e^{i \theta_j}\sqrt{\lambda_j} \ket{e_j}$ an arbitrary normalized state with $\lambda_0 > 0$, $\lambda_j \geq 0 \ \forall j \geq 1$, and $\sum_j \lambda_j = 1$. The following states, together with $\ket{\psi_0}$, form another basis:
\begin{multline}
\label{eq:APPCanonicalBasis}
\ket{\psi_k} =
\prtb{\gamma_{k}\ket{e_k} -e^{i(\theta_0 -\theta_k)} \sqrt{\lambda_0 \lambda_k} \ket{e_0} +
\\- \sum_{j=k+1}^{N-1} e^{i(\theta_j -\theta_k)}  \sqrt{\lambda_k \lambda_j} \ket{e_j}}/
{\sqrt{\gamma_k \gamma_{k-1}}},
\end{multline}
where $1\leq k \leq N-1$, $\gamma_k = 1 - \sum_{j=1}^{k} \lambda_j  = \lambda_0 + \sum_{j=k+1}^{N-1} \lambda_j$, and $\gamma_0 = 1$.
\end{theorem}
\begin{proof}
First we prove that any two states of the set are orthogonal.
We start by considering $n > k \geq 1$:
\begin{multline}
\ip{\psi_n}{\psi_k} \propto  e^{-i \prt{\theta_k - \theta_n}} \sqrt{\lambda_k \lambda_n} \prt{ \lambda_0
-\gamma_n +\sum_{i=n+1}^{N-1} \lambda_i }
\\=
e^{- i \prt{\theta_k - \theta_n}} \sqrt{\lambda_k \lambda_n} \prt{ \lambda_0
-\gamma_n + \gamma_n - \lambda_0 }=0.
\end{multline}
Indeed, the same holds for $ 1 \leq n < k$.
In the same manner, it is easy to show that $\ip{\psi_0}{\psi_k}=0 \ \forall k$.
Lastly, we prove that the states are normalized:
\begin{multline}
\gamma_k \gamma_{k-1} \ip{\psi_k}
=
\gamma_k^2 + \lambda_k \prt{\lambda_0 + \sum_{i=k+1}^{N-1}\lambda_i}
\\=
\gamma_k^2 + \lambda_k \gamma_k
=
\gamma_k \prt{\gamma_k + \lambda_k}
=
\gamma_k \gamma_{k-1}.
\end{multline}
\end{proof}
We remark that when $\lambda_k=0$, $\ket{\psi_k}=\ket{e_k}$.

Theorem~\ref{th:OrthogonalBasis} lets us write down a unitary operator connecting two arbitrary states.
Suppose we want to obtain state $\ket{\psi'}$ from state $\ket{\psi}$.
As first step, we define the operator
\begin{equation}
\label{eq:APPUnitaryOperatorGeneral}
U (\ket{\psi}) = \sum_{i=0}^{N-1} \dyad{\psi_i}{e_i},
\end{equation}
which maps the state $\ket{e_0}$ to $\ket{\psi_0}=\ket{\psi}$.
If $\lambda_0 = 0$, the state $\ket{\psi_{N-1}}$ (and possibly others) cannot be properly defined since, e.g., $\gamma_{N-1}=0$. However, we can write the target state with respect to the basis $\{\ket{\tilde{e}_i}\}_{i=0}^{N-1}$ in which $\ket{\tilde{e}_{i^*}} = \ket{e_0}$ and $\ket{\tilde{e}_{0}} = \ket{e_{i^*}}$, where $i^*$ is any index value such that $\lambda_{i^*} > 0$, while for $i\neq 0, \, i^*$, $\ket{\tilde{e}_i}=\ket{e_i}$. Accordingly, we write the states of Theorem~\ref{th:OrthogonalBasis} and the corresponding new $\tilde{U} (\ket{\psi})$ with respect to this new basis. Then, before $\tilde{U} (\ket{\psi})$, we apply another unitary operator, $U_{\mathrm{SWAP}}$, which swaps the states $\ket{e_0}$ and $\ket{e_{i^*}}$. In this way, when $\lambda_0=0$, we can define $U (\ket{\psi})$ as  $U (\ket{\psi}) =\tilde{U} (\ket{\psi})U_{\mathrm{SWAP}} $, thus assuring the transition $\ket{e_0} \rightarrow \ket{\psi}$.

By applying the same reasoning to $\ket{\psi'}$, i.e., $\ket{\psi'} = U (\ket{\psi'}) \ket{e_0}$, with the proper modification if for the state $\ket{\psi'}$, $\lambda_0 = 0$, one gets
\begin{equation}
\ket{\psi'} = U (\ket{\psi'}) U^\dagger (\ket{\psi}) \ket{\psi}. 
\end{equation}

The  transition from $\ket{e_0}$ to $\ket{\psi} $ can be obtained in another way, i.e., by exploiting the composition of elementary two-dimensional rotations.
We can write
\begin{equation}
\label{APPeq:UnitaryR}
U_R \prt{\ket{\psi}} \equiv \prod_{j=1}^{N} U_j \prt{\ket{\psi}},
\end{equation}
where, for $1 \leq j < N$ (hereafter we drop the argument $\ket{\psi}$ to lighten the notation),
\begin{multline}
U_j = \sqrt{\frac{\gamma_j}{\gamma_{j-1}}} \prt{\dyad{e_0}+\dyad{e_j}} +\sum_{k \neq 0,j} \dyad{e_k} \\
+\sqrt{\frac{\lambda_j}{\gamma_{j-1}}} \prt{e^{i \theta_j}\dyad{e_j}{e_0}- e^{- i \theta_j}\dyad{e_0}{e_j}},
\end{multline}
while, for $j=N$,
\begin{equation}
U_N = e^{i \theta_0} \dyad{e_0} + \sum_{k=1}^{N-1} \dyad{e_k}.
\end{equation}
To compute the action of $U_R$ we first compute the action of its constituents ($1 \leq  j < N$) for $k=0$ and $k>0$:
\begin{equations}
U_j \ket{e_0} &= \sqrt{\frac{\gamma_j}{\gamma_{j-1}}}  \ket{e_0} + e^{i \theta_j}\sqrt{\frac{\lambda_j}{\gamma_{j-1}}} \ket{e_j},\\
U_j \ket{e_k} &= 
\begin{cases}
\sqrt{\frac{\gamma_k}{\gamma_{k-1}}}  \ket{e_k} -e^{-i \theta_k}\sqrt{\frac{\lambda_k}{\gamma_{k-1}}} \ket{e_0}, &\quad j=k,\\
\ket{e_k},&\quad j \neq k.
\end{cases}
\end{equations}

Let us start by studying the action of $U_R$ on the state $\ket{e_0}$. By applying the first $n-1$ matrices $U_j$ [i.e., with $j$ going from $1$ to $n-1$ in Eq.~\eqref{APPeq:UnitaryR}], the amplitude of $\ket{e_0}$ goes to $\sqrt{\gamma_{n-1}}$ while that of $\ket{e_n}$ remains zero. Then, the application of the \blue{matrix $U_n$} makes the amplitude of $\ket{e_0}$ go to $\sqrt{\gamma_n}$ and that of $\ket{e_n}$ from zero to 
\begin{equation}
\sqrt{\gamma_{n-1}} e^{i \theta_n } \sqrt{\frac{\lambda_n}{\gamma_{n-1}}} = e^{i \theta_n} \sqrt{\lambda_n},
\end{equation}
which is the final amplitude of $\ket{e_n}$ in the state $\ket{\psi}$, and will remain unvaried by the action of the subsequent unitaries $U_j$ with $j > n$.
The application of all unitaries but $U_N$ brings the amplitude of $\ket{e_0}$ to $\sqrt{\gamma_{N-1}} = \sqrt{\lambda_0}$.
Therefore, since $U_N$ correctly changes the phase of $\ket{e_0}$, it follows that $U_R \ket{e_0} = \ket{\psi}$.

We can also calculate $U_R \ket{e_k}$. 
First, we notice that
\begin{equation}
U_R \ket{e_k} = \prod_{j = k}^{N} U_j \ket{e_k},
\end{equation}
since the application of the first $k-1$ matrices make no effect.
Then, 
\begin{multline}
\prod_{j =k}^{N-1} U_j \ket{e_k} \\=
\prod_{j =k+1}^{N-1} U_j \prt{
\sqrt{\frac{\gamma_k}{\gamma_{k-1}}}  \ket{e_k} -e^{- i \theta_j}\sqrt{\frac{\lambda_k}{\gamma_{k-1}}} \ket{e_0}} \\=
\sqrt{\frac{\gamma_k}{\gamma_{k-1}}}\ket{e_k} -e^{ - i\theta_k}\sqrt{\frac{\lambda_k}{\gamma_{k-1}}}  \prod_{j =k+1}^{N-1} U_j \ket{e_0}.
\end{multline}
This last term can be computed too:
\begin{multline}
\prod_{j =k+1}^{N-1} U_j \ket{e_0} \\= \prod_{j=k+2}^{N-1} U_j \prt{\sqrt{\frac{\gamma_{k+1}}{\gamma_{k}}}  \ket{e_0} + e^{i \theta_{k+1}}\sqrt{\frac{\lambda_{k+1}}{\gamma_{k}}} \ket{e_{k+1}}} \\=
\sqrt{\frac{\gamma_{N-1}}{\gamma_k}} \ket{e_0} + \sum_{j=k+1}^{N-1} e^{i \theta_{j}}\sqrt{\frac{\lambda_{j}}{\gamma_{k}}} \ket{e_{j}}.
\end{multline}
Eventually, we get (also using $\gamma_{N-1} = \lambda_0$)
\begin{multline}
U_R \ket{e_k} =
\sqrt{\frac{\gamma_k}{\gamma_{k-1}}}\ket{e_k} -e^{i \prt{\theta_0 - \theta_k}}\sqrt{\frac{\lambda_0 \lambda_k}{\gamma_k \gamma_{k-1}}} \ket{e_0} +
\\
- \sum_{j=k+1}^{N-1} e^{ i \prt{\theta_j - \theta_{k}}}\sqrt{\frac{\lambda_k \lambda_{j}}{\gamma_{k}\gamma_{k-1}}} \ket{e_{j}},
\end{multline}
which is equal to Eq.~\eqref{eq:APPCanonicalBasis}.
Therefore, we have just shown how the matrix $U\prt{\ket{\psi}}$ is explicitly decomposable as a product of elementary two-dimensional rotations, i.e., $U \prt{\ket{\psi}}=U_R \prt{\ket{\psi}}$.

\section{\label{sec:APPMinimizationSimple}Minimization of the expended energy in the modified simple thermalization protocol}

Here, we show how to find the states minimizing $\Eexp$ for a fixed entanglement $\E$ in the modified simple approach of Sec.~\ref{subsec:SimpleApproach} using the interaction Hamiltonian of Eq.~\eqref{eq:SimpleInteractionHamiltonianModified}.
Since this Hamiltonian can only generate states belonging to the subspace generated by the kets $\ket{E_i}$, any target state has the form of Eq.~\eqref{eq:TargetStates}.
Therefore, we can rewrite the formula of $\Eexp$ [cfr. Eq.~\eqref{eq:ProtocolEfficiency2}] as follows (setting $E_0 = 0$ for simplicity of notation):
\begin{equation}
\Eexp = \sum_{i=1}^{N_A-1}\lambda_i E_i +V_M \prt{1-\lambda_0} = \sum_{i=1}^{N_A-1}\lambda_i \prt{E_i+V_M}.
\end{equation}
The minimization problem can be now easily solved by considering the coefficients $\lambda_i$ as the coefficients of a thermal state with respect to a  fictitious Hamiltonian with eigenvalues $\tilde{E}_i = E_i + V_M$ and $\tilde{E}_0 = 0$ (see also Sec.~I and Sec.~II of the Supplemental Material of Ref.~\cite{Piccione2020Energy}). 
By doing so we get
\begin{equation}
\tilde{\lambda}_i = \frac{e^{-\beta_c \tilde{E}_i}}{\sum_{j=0}^{N_A-1} e^{-\beta_c \tilde{E}_j}},
\end{equation}
where $\beta_c$ is obtained by requiring that
\begin{equation}
-\sum_{i=0}^{N_A-1} \tilde{\lambda}_i \ln \tilde{\lambda}_i = \E.
\end{equation}

\section{\label{sec:APPExplicitFormHI} Construction of the interaction matrix obtained through a unitary transformation}

In the zero-temperature limit, dissipation can be used to generate pure states.
If a system is described by a Hamiltonian $H_0$ with a \blue{non}degenerate lowest eigenvalue, dissipation due to the interaction with a zero-temperature bath will generally lead to the ground state of the system~\cite{BookBreuer2002}.
Let us call $\prtg{e_i}_{i=0}^{N-1}$ an eigenbasis for $H_0$. For any desired state $\ket{\psi}$ (for simplicity, we only consider states with $\lambda_0>0$),
by using a proper unitary operator $U (\ket{\psi})$ defined as in Eq.~\eqref{eq:APPUnitaryOperatorGeneral}, it is possible to write  an Hermitian operator $H' (\ket{\psi}) = U (\ket{\psi}) H_0 U^\dagger (\ket{\psi})$ such that $\ket{\psi}$ is the ground state of $H'$:
\begin{equation}
\label{eq:APPNewHamiltonian}
H' (\ket{\psi}) = U (\ket{\psi}) H_0 U^\dagger (\ket{\psi})= 
\sum_{\bl=0}^{N-1} e_\bl \dyad{\psi_\bl},
\end{equation}
where \blue{$e_l$} runs over all the eigenvalues of $H_0$, i.e., with respect to the notation of Sec.~\ref{sec:BackgroundInformation}, $e_\bl$ is a shorthand for all combinations of $A_j + B_k$, for all $A_j$ and $B_k$, and with $e_0 = E_0 = A_0 + B_0$.
Writing $H' = H_0 + H_I$ (we drop the explicit dependence on $\ket{\psi}$ to lighten the notation), we obtain
\begin{equation}
\label{eq:APPInteractionHamiltonian}
H_I = \sum_{\bl=0}^{N-1} e_\bl \prtb{ \dyad{\psi_\bl} - \dyad{e_\bl}}.
\end{equation}

We can explicitly find the coefficients of the Hamiltonian $H_I$ by calculating $U (\ket{\psi})_{j \bl} = \mel{e_j}{U (\ket{\psi})}{e_\bl} =  \ip{e_j}{\psi_\bl}$ [cfr. Eq.~\eqref{eq:APPUnitaryOperatorGeneral}].
For $\bl \geq 1$,  using Eq.~\eqref{eq:APPCanonicalBasis}, we get:
\begin{equation}
\label{eq:APPInternalProduct}
\ip{e_j}{\psi_\bl} =
\begin{cases}
-e^{i (\theta_0 - \theta_\bl)}\sqrt{\frac{\lambda_0 \lambda_\bl}{\gamma_\bl \gamma_{\bl-1}}}, \quad &j=0,\\
0, \quad &0<j <\bl,\\
\sqrt{\frac{\gamma_\bl}{\gamma_{\bl-1}}}, \quad &j=\bl,\\
-e^{ i (\theta_j - \theta_\bl)} \sqrt{\frac{\lambda_j \lambda_\bl}{\gamma_\bl \gamma_{\bl-1}}}, \quad &j >\bl .
\end{cases}
\end{equation}
Let us start from the diagonal elements. For $j=0$ we get
\begin{equation}
\ev{H_I}{e_0} = e_0 (\lambda_0 - 1) + \lambda_0 \sum_{\bl=1}^{N-1} \frac{\lambda_\bl e_\bl}{\gamma_\bl \gamma_{\bl-1}},
\end{equation}
while for $j>0$ we have
\begin{equation}
\ev{H_I}{e_j}=
\lambda_j \prt{ e_0 - \frac{e_j}{\gamma_{j-1}} + \sum_{\bl=1}^{j-1} \frac{\lambda_\bl e_\bl}{\gamma_\bl \gamma_{\bl-1}}}.
\end{equation}

Since $H_I$ is an Hermitian matrix, it suffices to continue the calculations by analyzing the $ j > k \ge 0$ case. For $j>k=0$ we get
\begin{equation}
\!\mel{e_0}{H_I}{e_j} = e^{i (\theta_0 - \theta_j)} \sqrt{\lambda_0 \lambda_j} \prt{\!e_0 - \frac{e_j}{\gamma_{j-1}} + \sum_{\bl=1}^{j-1} \frac{\lambda_\bl e_\bl}{\gamma_\bl \gamma_{\bl-1}}\!} \!,
\end{equation}
\blue{whereas} for $j> k >0$ we have
\begin{align}
&\mel{e_k}{H_I}{e_j} =
\sum_{\bl=0}^{N-1} e_\bl \ip{e_k}{\psi_\bl} \ip{\psi_\bl}{e_j}
\nonumber \\&=
e^{i (\theta_k - \theta_j)}\sqrt{\lambda_k \lambda_j} \prt{ e_0 - \frac{e_k}{\gamma_{k-1}} + 
	\sum_{\bl=1}^{k-1} \frac{\lambda_\bl e_\bl}{\gamma_\bl \gamma_{\bl-1}}}.
\end{align}

By focusing on the case in which the target states are like those of Eq.~\eqref{eq:TargetStates}, we get Eq.~\eqref{eq:InteractionMatrix}.

\section{\label{sec:APPEfficiencyUnitaryGlobal} Efficiency and expended energy for the \nonloc{} unitary approach}

In this Appendix, we obtain the general formulas for the efficiency and expended energy of the \nonloc{} unitary approach by using their definition of Eq.~\eqref{eq:ProtocolEfficiency}.

By using Eqs.~\eqref{eq:APPNewHamiltonian}, \eqref{eq:APPInteractionHamiltonian}, and \eqref{eq:APPInternalProduct} in Eq.~\eqref{eq:ProtocolEfficiency}, we get
\begin{equation}
\label{eq:APPEnergyEfficiency}
\eta 
= \frac{e_0 \prt{\lambda_0 - 1} + \sum_{i=1}^{N-1} \lambda_i e_i }{2 e_0 (\lambda_0 - 1) + \sum_{i=1}^{N-1} \lambda_i e_i \prtq{1+\lambda_0/\prt{\gamma_i \gamma_{i-1}}}}.
\end{equation}
If we set $e_0=0$, which can always be done by adding a constant term to $H_0$, we get
\begin{equation}
\label{APPeq:GlobalUnitaryEfficiency}
\eta = \frac{\sum_{i=1}^{N-1} \lambda_i e_i}{\sum_{i=1}^{N-1} \lambda_i e_i \prtq{1+\lambda_0/\prt{\gamma_i \gamma_{i-1}}}},
\end{equation}
which makes clear that $\eta<1$, as it should be.
When considering target states like those of Eq.~\eqref{eq:TargetStates}, the efficiency of Eq.~\eqref{APPeq:GlobalUnitaryEfficiency} and the denominator therein appearing become, respectively, the efficiency and the energy expense reported in Eq.~\eqref{eq:EfficiencyGlobalUnitary}.

Now, we prove that $\Eexp^S$ [see Eq.~\eqref{eq:ProtocolEfficiency}] cannot be higher than $2 [e_{N-1}-e_0]$ in the \nonloc{} unitary approach:
\begin{multline}
\Eexp^S =  \ev{H_I}{e_0} - \ev{H_I}{\psi} = \\
\ev{\left(H_I+H_0 - H_0\right)}{e_0} - \ev{\left(H_I + H_0 - H_0\right)}{\psi} =\\
\ev{U H_0 U^\dagger}{e_0} - e_0 - \ev{U H_0 U^\dagger}{\psi} +\ev{H_0}{\psi}  =\\
\ev{\left(U H_0 U^\dagger+U^\dagger H_0 U\right)}{e_0} - 2 e_0 \leq 
2 e_{N-1} - 2 e_0,
\end{multline}
where the last inequality comes from the invariance of the spectrum for unitary transformations.

\section{Efficiency and expended energy in the \loc{} unitary approaches\label{sec:APPEfficiencyLocalUnitary}}

Here, we want to calculate the quantities of Eq.~\eqref{eq:ProtocolEfficiency} in the case of the states obtainable within the MSSG unitary approaches, i.e., states $\ket{\phi}$ of the form of Eq.~\eqref{eq:TargetStates}. In these approaches, a state $\ket{\phi}$ is obtained by applying $U_A$ or $U_B$ to $\ket{E_0}$  and then the corresponding generalized CNOT gate [see Eqs.~\eqref{eq:CNOT} and \eqref{eq:CNOTB}]  is applied to the resulting state on system $S$. Let us focus on the $U_A$ case, where
\begin{equation}
H_I = U_{G_A} U_A H_0 U_A^\dagger U_{G_A}^\dagger - H_0.
\end{equation}
First, we notice that $U_{G_A} \ket{A_0 B_i} = U_{G_A}^\dagger \ket{A_0 B_i} = \ket{A_0 B_i}$.
Therefore, using also Eqs.~\eqref{eq:UnitaryOperatorLocalA} and \eqref{eq:APPInternalProduct} applied to the calculation of $\ip{A_i}{\psi_k^A}$, we obtain
\begin{align} 
\ev{H_I}{E_0} 
&= \ev{U_A H_A U_A^\dagger }{A_0} + B_0 - E_0,\nonumber \\
&=  A_0 \prt{\lambda_0-1} +\lambda_0\sum_{i=1}^{N_A-1}\frac{\lambda_i A_i}{\gamma_i \gamma_{i-1}} .
\end{align}
Then, we calculate
\begin{equation}
\ev{H_I}{\phi} = E_0 \prt{1-\lambda_0} - \sum_{i=1}^{N_A-1} \lambda_i E_i.
\end{equation}
Eventually, the expended energy in this case is equal to
\begin{equation}
\label{eq:APPEnergyExpenseA}
\Eexp^A = \sum_{i=1}^{N_A-1} \prtq{\lambda_i \prt{ E_i + \frac{\lambda_0 A_i}{\gamma_i \gamma_{i-1}}} } - \prt{A_0 + E_0} \prt{1-\lambda_0}.
\end{equation}

Now, we show that the expended energy within the global approach is higher that in the MSSG unitary approaches when they generate the same state $\ket{\phi}$. The expended energy in the \nonloc{} approach can be written as [cfr. Eq.~\eqref{eq:APPEnergyEfficiency} applied to a state $\ket{\phi}$]
\begin{equation}
\Eexp^S = \sum_{i=1}^{N_A-1} \lambda_i E_i \prtq{1+\frac{\lambda_0}{\gamma_i \gamma_{i-1}}} - 2 E_0 \prt{1 -\lambda_0},
\end{equation}
and the difference between $\Eexp^S$ and $\Eexp^A$ is equal to
\begin{equation}
\Eexp^S - \Eexp^A = \lambda_0 \sum_{i=1}^{N_A-1} \frac{\lambda_i B_i}{\gamma_i \gamma_{i-1}} - B_0 \prt{1-\lambda_0}.
\end{equation}
Since we can set $B_0=0$ without changing the above quantity, we conclude that the \loc{} version of the protocol \blue{based on $\tilde{U}_A$} always performs better than the \nonloc{} one.
The derivation of $\Eexp^B$ is similar since, analogously to the previous case, $U_{G_B} \ket{A_i B_0} = U_{G_B}^\dagger \ket{A_i B_0} = \ket{A_i B_0}$, and leads to the same expression of Eq.~\eqref{eq:APPEnergyExpenseA} with the quantities  $A_i$ and $B_i$ swapped between them.

To complete the analysis we also write the difference regarding the expended energy in the two \loc{} approaches:
\begin{equation}
\Eexp^B - \Eexp^A = \lambda_0 \sum_{i=1}^{N_A-1} \frac{\lambda_i \prt{B_i-A_i} }{\gamma_i \gamma_{i-1}} - \prt{B_0-A_0} \prt{1-\lambda_0}.
\end{equation}
In this case, which approach is better depends both on the target state and the spectra of $H_A$ and $H_B$.

\section{\label{sec:TwoQubits}Explicit calculations for two qubits}

In this Appendix, we explicitly calculate the three unitary transformations proposed in Sec.~\ref{sec:UnitaryTransformation} and the interaction Hamiltonians proposed in Sec.~\ref{sec:TwoApproaches} for the archetypical example of a bipartite system composed of two qubits, having as free  Hamiltonians
\begin{equation}
H_A = \frac{\hbar \omega_A}{2} \sigma_z^{A},
\qquad
H_B = \frac{\hbar \omega_B}{2} \sigma_z^{B},
\end{equation}
where $\sigma_z^{X}$ is the $z$-Pauli matrix of qubit $X$ ($X=A, B$). 
As basis for $H_0$ we use $\{\ket{00},\ket{01},\ket{10},\ket{11}\}$, where, for each qubit, $\sigma_z^{X} \ket{1} = \ket{1}$ and $\sigma_z^{X}\ket{0} = - \ket{0}$.
Let us define $\omega =\prt{\omega_A + \omega_B}/2$, then the MEES for a two qubit system, as reported in \cite{Piccione2020Energy}, and considering arbitrary phases, is given by
\begin{equation}
\stMin= e^{i \theta_0}\sqrt{\lambda} \ket{00} + e^{i \theta_1}\sqrt{1-\lambda}\ket{11},
\end{equation}
where $1/\lambda=1+e^{-2\beta \hbar \omega}$.

Starting from Eqs.~\eqref{eq:UnitaryOperatorGlobal}, \eqref{eq:UnitaryOperatorLocalA}, \eqref{eq:CNOT}, and \eqref{eq:CNOTB} [$U_B $ is defined analogously to  what done in Eq.~\eqref{eq:UnitaryOperatorLocalA}], straightforward calculations give
for $U_S$, $\tilde{U}_A = U_{G_A} U_A$, and $\tilde{U}_B = U_{G_B} U_B$:
\begin{equation}
U_S = \mqty( 
e^{i \theta_0}\sqrt{\lambda} & 0 & 0 & -e^{i (\theta_0 -\theta_1)}\sqrt{1-\lambda} \\
0 & 1 & 0 & 0 \\
0 & 0 & 1 & 0 \\
e^{i \theta_1}\sqrt{1-\lambda} & 0 & 0 & \sqrt{\lambda}
),
\end{equation}
\begin{widetext}
\begin{equation}
\tilde{U}_A = \mqty( 
e^{i \theta_0}\sqrt{\lambda} & 0 &  -e^{i (\theta_0-\theta_1)}\sqrt{1-\lambda} & 0 \\
0 & e^{i \theta_0}\sqrt{\lambda} & 0 & -e^{i (\theta_0-\theta_1)}\sqrt{1-\lambda} \\
0 & e^{i \theta_1}\sqrt{1-\lambda} & 0 & \sqrt{\lambda} \\
e^{i \theta_1}\sqrt{1-\lambda} & 0 & \sqrt{\lambda} & 0
),
\end{equation}
and
\begin{equation}
\tilde{U}_B = \mqty( 
e^{i \theta_0}\sqrt{\lambda} & -e^{i (\theta_0-\theta_1)}\sqrt{1-\lambda} & 0 & 0\\
0 & 0 & e^{i \theta_1}\sqrt{1-\lambda} & \sqrt{\lambda} \\
0 & 0 & e^{i \theta_0}\sqrt{\lambda} & -e^{i (\theta_0-\theta_1)}\sqrt{1-\lambda} \\
e^{i \theta_1}\sqrt{1-\lambda} & \sqrt{\lambda} & 0 & 0
).
\end{equation}
\end{widetext}
We can see, explicitly, that the three unitary operators proposed in Sec.~\ref{sec:UnitaryTransformation} are different.

Regarding the thermalization approach, we can immediately write down the interaction Hamiltonian obtained through the simple approach:
\begin{widetext}
\begin{equation}
\Hsi_I = \mqty( 
\hbar \omega -V_S \lambda & 0 & 0 & - e^{i (\theta_0-\theta_1)}V_S \sqrt{\lambda\prt{1-\lambda}}\\
0 & \hbar \delta_A& 0 & 0 \\
0 & 0 & \hbar \delta_B & 0 \\
-e^{-i (\theta_0-\theta_1)}V_S \sqrt{\lambda\prt{1-\lambda}} & 0 & 0 & -V_S (1-\lambda) - \hbar \omega
),
\end{equation}
where $\delta_A = \prt{\omega_A - \omega_B}/2$ and $\delta_B = \prt{\omega_B - \omega_A}/2$.
For comparison, the modified simple approach leads to
\begin{equation}
\label{eq:APPInteractionSimpleModifiedQubits}
\Hsim = \mqty( 
\hbar \omega -V_M  \lambda & 0 & 0 & -e^{i (\theta_0-\theta_1)} V_M \sqrt{\lambda\prt{1-\lambda}}\\
0 & 0 & 0 & 0 \\
0 & 0 & 0 & 0 \\
- e^{-i (\theta_0-\theta_1)} V_M \sqrt{\lambda\prt{1-\lambda}} & 0 & 0 & -V_M (1-\lambda) - \hbar \omega
),
\end{equation}
which lacks the diagonal terms $\hbar \delta_A$ et $\hbar \delta_B$.

Turning to the zero-temperature approach based on the \nonloc{} unitary transformation, the use of $U_S$ leads to
\begin{equation}
\label{eq:APPInteractionUnitaryGlobalQubits}
H_I^{S}  = 2 \hbar \omega \mqty( 
1-\lambda & 0 & 0 & -e^{i (\theta_0-\theta_1)}\sqrt{\lambda\prt{1-\lambda}} \\
0 & 0 & 0 & 0 \\
0 & 0 & 0 & 0 \\
-e^{-i (\theta_0-\theta_1)}\sqrt{\lambda\prt{1-\lambda}} & 0 & 0 & \lambda -1
).
\end{equation}
\end{widetext}
Even if this Hamiltonian is similar to that of Eq.~\eqref{eq:APPInteractionSimpleModifiedQubits}, it is not the same one.
This was predictable since the modified simple approach changes the spectrum of $H_0$ \blue{whereas} all unitary ones do not.
Moreover, it does not exist a number $\blue{\mu}$ such that $\Hsim = \blue{\mu} H_I^S$ so that the two Hamiltonians cannot describe the same physics even with a rescaling of the energy.

Moving to the MSSG thermalization approaches, the use of $\tilde{U}_A$ leads to
\begin{widetext}
\begin{equation}
H_I^{A}  =\hbar \omega_A \mqty( 
1-\lambda & 0 & 0 & -e^{i (\theta_0-\theta_1)}\sqrt{\lambda\prt{1-\lambda}} \\
0 & 1-\lambda & -e^{i (\theta_0-\theta_1)}\sqrt{\lambda\prt{1-\lambda}} & 0 \\
0 & -e^{-i (\theta_0-\theta_1)}\sqrt{\lambda\prt{1-\lambda}} & \lambda-1+\frac{\omega_B}{\omega_A} & 0 \\
-e^{-i (\theta_0-\theta_1)}\sqrt{\lambda\prt{1-\lambda}} & 0 & 0 & \lambda -1-\frac{\omega_B}{\omega_A}
),
\end{equation}
\blue{whereas} the use of $\tilde{U}_B$ to
\begin{equation}
H_I^{B}  =\hbar \omega_B \mqty( 
1-\lambda & 0 & 0 & -e^{i (\theta_0-\theta_1)}\sqrt{\lambda\prt{1-\lambda}} \\
0 & \lambda-1+\frac{\omega_A}{\omega_B} & -e^{-i (\theta_0-\theta_1)}\sqrt{\lambda\prt{1-\lambda}} & 0 \\
0 & -e^{i (\theta_0-\theta_1)}\sqrt{\lambda\prt{1-\lambda}} & 1-\lambda & 0 \\
-e^{-i (\theta_0-\theta_1)}\sqrt{\lambda\prt{1-\lambda}} & 0 & 0 & \lambda -1-\frac{\omega_A}{\omega_B}
).
\end{equation}
\end{widetext}
Since all of these interaction Hamiltonians are different (even considering rescaling), we can conclude that the five methods give rise to physically different interactions.

In Ref.~\cite{Piccione2020Energy} the following interaction Hamiltonian was proposed to make $\stMin$ the ground state
\begin{equation}
H_I^{\textup{emp}} = \frac{\hbar g}{2} \prt{\sigma_x^A \sigma_x^B - \sigma_y^A \sigma_y^B},
\end{equation}
where $\sigma_i^X$ are Pauli operators on system $X$ with $i=x,y$.
The matrix structure of this operator is
\begin{equation}
\label{eq:APPInteractionEmpiricQubits}
H_I^{\textup{emp}} = \hbar g \mqty( 
0 & 0 & 0 & 1 \\
0 & 0 & 0 & 0 \\
0 & 0 & 0 & 0 \\
1 & 0 & 0 & 0 ),
\end{equation}
which lacks the diagonal terms present in all the general approaches we provided.
Therefore, the methods that we presented in this paper cannot be used to obtain this empirical Hamiltonian that we found by hand in Ref.~\cite{Piccione2020Energy}.

%\bibliography{mybib}

\begin{thebibliography}{47}%
	\makeatletter
	\providecommand \@ifxundefined [1]{%
		\@ifx{#1\undefined}
	}%
	\providecommand \@ifnum [1]{%
		\ifnum #1\expandafter \@firstoftwo
		\else \expandafter \@secondoftwo
		\fi
	}%
	\providecommand \@ifx [1]{%
		\ifx #1\expandafter \@firstoftwo
		\else \expandafter \@secondoftwo
		\fi
	}%
	\providecommand \natexlab [1]{#1}%
	\providecommand \enquote  [1]{``#1''}%
	\providecommand \bibnamefont  [1]{#1}%
	\providecommand \bibfnamefont [1]{#1}%
	\providecommand \citenamefont [1]{#1}%
	\providecommand \href@noop [0]{\@secondoftwo}%
	\providecommand \href [0]{\begingroup \@sanitize@url \@href}%
	\providecommand \@href[1]{\@@startlink{#1}\@@href}%
	\providecommand \@@href[1]{\endgroup#1\@@endlink}%
	\providecommand \@sanitize@url [0]{\catcode `\\12\catcode `\$12\catcode
		`\&12\catcode `\#12\catcode `\^12\catcode `\_12\catcode `\%12\relax}%
	\providecommand \@@startlink[1]{}%
	\providecommand \@@endlink[0]{}%
	\providecommand \url  [0]{\begingroup\@sanitize@url \@url }%
	\providecommand \@url [1]{\endgroup\@href {#1}{\urlprefix }}%
	\providecommand \urlprefix  [0]{URL }%
	\providecommand \Eprint [0]{\href }%
	\providecommand \doibase [0]{https://doi.org/}%
	\providecommand \selectlanguage [0]{\@gobble}%
	\providecommand \bibinfo  [0]{\@secondoftwo}%
	\providecommand \bibfield  [0]{\@secondoftwo}%
	\providecommand \translation [1]{[#1]}%
	\providecommand \BibitemOpen [0]{}%
	\providecommand \bibitemStop [0]{}%
	\providecommand \bibitemNoStop [0]{.\EOS\space}%
	\providecommand \EOS [0]{\spacefactor3000\relax}%
	\providecommand \BibitemShut  [1]{\csname bibitem#1\endcsname}%
	\let\auto@bib@innerbib\@empty
	%</preamble>
	\bibitem [{\citenamefont {Einstein}\ \emph {et~al.}(1935)\citenamefont
		{Einstein}, \citenamefont {Podolsky},\ and\ \citenamefont
		{Rosen}}]{Einstein1935}%
	\BibitemOpen
	\bibfield  {author} {\bibinfo {author} {\bibfnamefont {A.}~\bibnamefont
			{Einstein}}, \bibinfo {author} {\bibfnamefont {B.}~\bibnamefont {Podolsky}},\
		and\ \bibinfo {author} {\bibfnamefont {N.}~\bibnamefont {Rosen}},\ }\bibfield
	{title} {\bibinfo {title} {Can {Q}uantum-{M}echanical {D}escription of
			{P}hysical {R}eality {B}e {C}onsidered {C}omplete?},\ }\href
	{https://doi.org/10.1103/PhysRev.47.777} {\bibfield  {journal} {\bibinfo
			{journal} {Phys. Rev.}\ }\textbf {\bibinfo {volume} {47}},\ \bibinfo {pages}
		{777} (\bibinfo {year} {1935})}\BibitemShut {NoStop}%
	\bibitem [{\citenamefont {Bell}(1964)}]{Bell1964}%
	\BibitemOpen
	\bibfield  {author} {\bibinfo {author} {\bibfnamefont {J.~S.}\ \bibnamefont
			{Bell}},\ }\bibfield  {title} {\bibinfo {title} {On the {E}instein {P}odolsky
			{R}osen paradox},\ }\href
	{https://doi.org/10.1103/PhysicsPhysiqueFizika.1.195} {\bibfield  {journal}
		{\bibinfo  {journal} {Phys. Phys. Fiz.}\ }\textbf {\bibinfo {volume} {1}},\
		\bibinfo {pages} {195} (\bibinfo {year} {1964})}\BibitemShut {NoStop}%
	\bibitem [{\citenamefont {Horodecki}\ \emph {et~al.}(2009)\citenamefont
		{Horodecki}, \citenamefont {Horodecki}, \citenamefont {Horodecki},\ and\
		\citenamefont {Horodecki}}]{Horodecki2009}%
	\BibitemOpen
	\bibfield  {author} {\bibinfo {author} {\bibfnamefont {R.}~\bibnamefont
			{Horodecki}}, \bibinfo {author} {\bibfnamefont {P.}~\bibnamefont
			{Horodecki}}, \bibinfo {author} {\bibfnamefont {M.}~\bibnamefont
			{Horodecki}},\ and\ \bibinfo {author} {\bibfnamefont {K.}~\bibnamefont
			{Horodecki}},\ }\bibfield  {title} {\bibinfo {title} {Quantum entanglement},\
	}\href {https://doi.org/10.1103/RevModPhys.81.865} {\bibfield  {journal}
		{\bibinfo  {journal} {Rev. Mod. Phys.}\ }\textbf {\bibinfo {volume} {81}},\
		\bibinfo {pages} {865} (\bibinfo {year} {2009})}\BibitemShut {NoStop}%
	\bibitem [{\citenamefont {Nielsen}\ and\ \citenamefont
		{Chuang}(2010)}]{BookNielsen2010}%
	\BibitemOpen
	\bibfield  {author} {\bibinfo {author} {\bibfnamefont {M.~A.}\ \bibnamefont
			{Nielsen}}\ and\ \bibinfo {author} {\bibfnamefont {I.~L.}\ \bibnamefont
			{Chuang}},\ }\href {https://doi.org/10.1017/CBO9780511976667} {\emph
		{\bibinfo {title} {Quantum Computation and Quantum Information: 10th
				Anniversary Edition}}}\ (\bibinfo  {publisher} {Cambridge University Press},\
	\bibinfo {year} {2010})\BibitemShut {NoStop}%
	\bibitem [{\citenamefont {Plenio}\ \emph {et~al.}(1999)\citenamefont {Plenio},
		\citenamefont {Huelga}, \citenamefont {Beige},\ and\ \citenamefont
		{Knight}}]{Plenio1999}%
	\BibitemOpen
	\bibfield  {author} {\bibinfo {author} {\bibfnamefont {M.~B.}\ \bibnamefont
			{Plenio}}, \bibinfo {author} {\bibfnamefont {S.~F.}\ \bibnamefont {Huelga}},
		\bibinfo {author} {\bibfnamefont {A.}~\bibnamefont {Beige}},\ and\ \bibinfo
		{author} {\bibfnamefont {P.~L.}\ \bibnamefont {Knight}},\ }\bibfield  {title}
	{\bibinfo {title} {Cavity-loss-induced generation of entangled atoms},\
	}\href {https://doi.org/10.1103/PhysRevA.59.2468} {\bibfield  {journal}
		{\bibinfo  {journal} {Phys. Rev. A}\ }\textbf {\bibinfo {volume} {59}},\
		\bibinfo {pages} {2468} (\bibinfo {year} {1999})}\BibitemShut {NoStop}%
	\bibitem [{\citenamefont {Braun}(2002)}]{Braun2002}%
	\BibitemOpen
	\bibfield  {author} {\bibinfo {author} {\bibfnamefont {D.}~\bibnamefont
			{Braun}},\ }\bibfield  {title} {\bibinfo {title} {Creation of entanglement by
			interaction with a common heat bath},\ }\href
	{https://doi.org/10.1103/PhysRevLett.89.277901} {\bibfield  {journal}
		{\bibinfo  {journal} {Phys. Rev. Lett.}\ }\textbf {\bibinfo {volume} {89}},\
		\bibinfo {pages} {277901} (\bibinfo {year} {2002})}\BibitemShut {NoStop}%
	\bibitem [{\citenamefont {Bellomo}\ \emph {et~al.}(2008)\citenamefont
		{Bellomo}, \citenamefont {Lo~Franco}, \citenamefont {Maniscalco},\ and\
		\citenamefont {Compagno}}]{Bellomo2008}%
	\BibitemOpen
	\bibfield  {author} {\bibinfo {author} {\bibfnamefont {B.}~\bibnamefont
			{Bellomo}}, \bibinfo {author} {\bibfnamefont {R.}~\bibnamefont {Lo~Franco}},
		\bibinfo {author} {\bibfnamefont {S.}~\bibnamefont {Maniscalco}},\ and\
		\bibinfo {author} {\bibfnamefont {G.}~\bibnamefont {Compagno}},\ }\bibfield
	{title} {\bibinfo {title} {Entanglement trapping in structured
			environments},\ }\href {https://doi.org/10.1103/PhysRevA.78.060302}
	{\bibfield  {journal} {\bibinfo  {journal} {Phys. Rev. A}\ }\textbf {\bibinfo
			{volume} {78}},\ \bibinfo {pages} {060302(R)} (\bibinfo {year}
		{2008})}\BibitemShut {NoStop}%
	\bibitem [{\citenamefont {Sarlette}\ \emph {et~al.}(2011)\citenamefont
		{Sarlette}, \citenamefont {Raimond}, \citenamefont {Brune},\ and\
		\citenamefont {Rouchon}}]{Sarlette2011}%
	\BibitemOpen
	\bibfield  {author} {\bibinfo {author} {\bibfnamefont {A.}~\bibnamefont
			{Sarlette}}, \bibinfo {author} {\bibfnamefont {J.~M.}\ \bibnamefont
			{Raimond}}, \bibinfo {author} {\bibfnamefont {M.}~\bibnamefont {Brune}},\
		and\ \bibinfo {author} {\bibfnamefont {P.}~\bibnamefont {Rouchon}},\
	}\bibfield  {title} {\bibinfo {title} {Stabilization of {N}onclassical
			{S}tates of the {R}adiation {F}ield in a {C}avity by {R}eservoir
			{E}ngineering},\ }\href {https://doi.org/10.1103/PhysRevLett.107.010402}
	{\bibfield  {journal} {\bibinfo  {journal} {Phys. Rev. Lett.}\ }\textbf
		{\bibinfo {volume} {107}},\ \bibinfo {pages} {010402} (\bibinfo {year}
		{2011})}\BibitemShut {NoStop}%
	\bibitem [{\citenamefont {Krauter}\ \emph {et~al.}(2011)\citenamefont
		{Krauter}, \citenamefont {Muschik}, \citenamefont {Jensen}, \citenamefont
		{Wasilewski}, \citenamefont {Petersen}, \citenamefont {Cirac},\ and\
		\citenamefont {Polzik}}]{Krauter2011}%
	\BibitemOpen
	\bibfield  {author} {\bibinfo {author} {\bibfnamefont {H.}~\bibnamefont
			{Krauter}}, \bibinfo {author} {\bibfnamefont {C.~A.}\ \bibnamefont
			{Muschik}}, \bibinfo {author} {\bibfnamefont {K.}~\bibnamefont {Jensen}},
		\bibinfo {author} {\bibfnamefont {W.}~\bibnamefont {Wasilewski}}, \bibinfo
		{author} {\bibfnamefont {J.~M.}\ \bibnamefont {Petersen}}, \bibinfo {author}
		{\bibfnamefont {J.~I.}\ \bibnamefont {Cirac}},\ and\ \bibinfo {author}
		{\bibfnamefont {E.~S.}\ \bibnamefont {Polzik}},\ }\bibfield  {title}
	{\bibinfo {title} {Entanglement {G}enerated by {D}issipation and {S}teady
			{S}tate {E}ntanglement of {T}wo {M}acroscopic {O}bjects},\ }\href
	{https://doi.org/10.1103/PhysRevLett.107.080503} {\bibfield  {journal}
		{\bibinfo  {journal} {Phys. Rev. Lett.}\ }\textbf {\bibinfo {volume} {107}},\
		\bibinfo {pages} {080503} (\bibinfo {year} {2011})}\BibitemShut {NoStop}%
	\bibitem [{\citenamefont {Bellomo}\ and\ \citenamefont
		{Antezza}(2013)}]{Bellomo2013}%
	\BibitemOpen
	\bibfield  {author} {\bibinfo {author} {\bibfnamefont {B.}~\bibnamefont
			{Bellomo}}\ and\ \bibinfo {author} {\bibfnamefont {M.}~\bibnamefont
			{Antezza}},\ }\bibfield  {title} {\bibinfo {title} {Steady entanglement out
			of thermal equilibrium},\ }\href
	{https://doi.org/10.1209/0295-5075/104/10006} {\bibfield  {journal} {\bibinfo
			{journal} {Europhys. Lett.}\ }\textbf {\bibinfo {volume} {104}},\ \bibinfo
		{pages} {10006} (\bibinfo {year} {2013})}\BibitemShut {NoStop}%
	\bibitem [{\citenamefont {Bellomo}\ and\ \citenamefont
		{Antezza}(2015)}]{Bellomo2015}%
	\BibitemOpen
	\bibfield  {author} {\bibinfo {author} {\bibfnamefont {B.}~\bibnamefont
			{Bellomo}}\ and\ \bibinfo {author} {\bibfnamefont {M.}~\bibnamefont
			{Antezza}},\ }\bibfield  {title} {\bibinfo {title} {Nonequilibrium
			dissipation-driven steady many-body entanglement},\ }\href
	{https://doi.org/10.1103/PhysRevA.91.042124} {\bibfield  {journal} {\bibinfo
			{journal} {Phys. Rev. A}\ }\textbf {\bibinfo {volume} {91}},\ \bibinfo
		{pages} {042124} (\bibinfo {year} {2015})}\BibitemShut {NoStop}%
	\bibitem [{\citenamefont {Bellomo}\ \emph {et~al.}(2017)\citenamefont
		{Bellomo}, \citenamefont {Lo~Franco},\ and\ \citenamefont
		{Compagno}}]{Bellomo2017}%
	\BibitemOpen
	\bibfield  {author} {\bibinfo {author} {\bibfnamefont {B.}~\bibnamefont
			{Bellomo}}, \bibinfo {author} {\bibfnamefont {R.}~\bibnamefont {Lo~Franco}},\
		and\ \bibinfo {author} {\bibfnamefont {G.}~\bibnamefont {Compagno}},\
	}\bibfield  {title} {\bibinfo {title} {\emph{N} identical particles and one
			particle to entangle them all},\ }\href
	{https://doi.org/10.1103/PhysRevA.96.022319} {\bibfield  {journal} {\bibinfo
			{journal} {Phys. Rev. A}\ }\textbf {\bibinfo {volume} {96}},\ \bibinfo
		{pages} {022319} (\bibinfo {year} {2017})}\BibitemShut {NoStop}%
	\bibitem [{\citenamefont {Luo}\ \emph {et~al.}(2017)\citenamefont {Luo},
		\citenamefont {Zou}, \citenamefont {Wu}, \citenamefont {Liu}, \citenamefont
		{Han}, \citenamefont {Tey},\ and\ \citenamefont {You}}]{Luo2017}%
	\BibitemOpen
	\bibfield  {author} {\bibinfo {author} {\bibfnamefont {X.-Y.}\ \bibnamefont
			{Luo}}, \bibinfo {author} {\bibfnamefont {Y.-Q.}\ \bibnamefont {Zou}},
		\bibinfo {author} {\bibfnamefont {L.-N.}\ \bibnamefont {Wu}}, \bibinfo
		{author} {\bibfnamefont {Q.}~\bibnamefont {Liu}}, \bibinfo {author}
		{\bibfnamefont {M.-F.}\ \bibnamefont {Han}}, \bibinfo {author} {\bibfnamefont
			{M.~K.}\ \bibnamefont {Tey}},\ and\ \bibinfo {author} {\bibfnamefont
			{L.}~\bibnamefont {You}},\ }\bibfield  {title} {\bibinfo {title}
		{Deterministic entanglement generation from driving through quantum phase
			transitions},\ }\href {https://doi.org/10.1126/science.aag1106} {\bibfield
		{journal} {\bibinfo  {journal} {Science}\ }\textbf {\bibinfo {volume}
			{355}},\ \bibinfo {pages} {620} (\bibinfo {year} {2017})}\BibitemShut
	{NoStop}%
	\bibitem [{\citenamefont {\c{C}akmak}\ \emph {et~al.}(2019)\citenamefont
		{\c{C}akmak}, \citenamefont {Campbell}, \citenamefont {Vacchini},
		\citenamefont {M\"ustecapl{\i}o\u{g}lu},\ and\ \citenamefont
		{Paternostro}}]{Cakmak2019}%
	\BibitemOpen
	\bibfield  {author} {\bibinfo {author} {\bibfnamefont {B.}~\bibnamefont
			{\c{C}akmak}}, \bibinfo {author} {\bibfnamefont {S.}~\bibnamefont
			{Campbell}}, \bibinfo {author} {\bibfnamefont {B.}~\bibnamefont {Vacchini}},
		\bibinfo {author} {\bibfnamefont {{\"O}.~E.}\ \bibnamefont
			{M\"ustecapl{\i}o\u{g}lu}},\ and\ \bibinfo {author} {\bibfnamefont
			{M.}~\bibnamefont {Paternostro}},\ }\bibfield  {title} {\bibinfo {title}
		{Robust multipartite entanglement generation via a collision model},\ }\href
	{https://doi.org/10.1103/PhysRevA.99.012319} {\bibfield  {journal} {\bibinfo
			{journal} {Phys. Rev. A}\ }\textbf {\bibinfo {volume} {99}},\ \bibinfo
		{pages} {012319} (\bibinfo {year} {2019})}\BibitemShut {NoStop}%
	\bibitem [{\citenamefont {Goldberg}\ and\ \citenamefont
		{James}(2019)}]{Goldberg2019}%
	\BibitemOpen
	\bibfield  {author} {\bibinfo {author} {\bibfnamefont {A.~Z.}\ \bibnamefont
			{Goldberg}}\ and\ \bibinfo {author} {\bibfnamefont {D.~F.~V.}\ \bibnamefont
			{James}},\ }\bibfield  {title} {\bibinfo {title} {Entanglement generation via
			diffraction},\ }\href {https://doi.org/10.1103/PhysRevA.100.042332}
	{\bibfield  {journal} {\bibinfo  {journal} {Phys. Rev. A}\ }\textbf {\bibinfo
			{volume} {100}},\ \bibinfo {pages} {042332} (\bibinfo {year}
		{2019})}\BibitemShut {NoStop}%
	\bibitem [{\citenamefont {Castellini}\ \emph {et~al.}(2019)\citenamefont
		{Castellini}, \citenamefont {Bellomo}, \citenamefont {Compagno},\ and\
		\citenamefont {Lo~Franco}}]{Bellomo2019}%
	\BibitemOpen
	\bibfield  {author} {\bibinfo {author} {\bibfnamefont {A.}~\bibnamefont
			{Castellini}}, \bibinfo {author} {\bibfnamefont {B.}~\bibnamefont {Bellomo}},
		\bibinfo {author} {\bibfnamefont {G.}~\bibnamefont {Compagno}},\ and\
		\bibinfo {author} {\bibfnamefont {R.}~\bibnamefont {Lo~Franco}},\ }\bibfield
	{title} {\bibinfo {title} {Activating remote entanglement in a quantum
			network by local counting of identical particles},\ }\href
	{https://doi.org/10.1103/PhysRevA.99.062322} {\bibfield  {journal} {\bibinfo
			{journal} {Phys. Rev. A}\ }\textbf {\bibinfo {volume} {99}},\ \bibinfo
		{pages} {062322} (\bibinfo {year} {2019})}\BibitemShut {NoStop}%
	\bibitem [{\citenamefont {Egger}\ \emph {et~al.}(2019)\citenamefont {Egger},
		\citenamefont {Ganzhorn}, \citenamefont {Salis}, \citenamefont {Fuhrer},
		\citenamefont {M\"uller}, \citenamefont {{Kl. Barkoutsos}}, \citenamefont
		{Moll}, \citenamefont {Tavernelli},\ and\ \citenamefont
		{Filipp}}]{Egger2019}%
	\BibitemOpen
	\bibfield  {author} {\bibinfo {author} {\bibfnamefont {D.~J.}\ \bibnamefont
			{Egger}}, \bibinfo {author} {\bibfnamefont {M.}~\bibnamefont {Ganzhorn}},
		\bibinfo {author} {\bibfnamefont {G.}~\bibnamefont {Salis}}, \bibinfo
		{author} {\bibfnamefont {A.}~\bibnamefont {Fuhrer}}, \bibinfo {author}
		{\bibfnamefont {P.}~\bibnamefont {M\"uller}}, \bibinfo {author}
		{\bibfnamefont {P.}~\bibnamefont {{Kl. Barkoutsos}}}, \bibinfo {author}
		{\bibfnamefont {N.}~\bibnamefont {Moll}}, \bibinfo {author} {\bibfnamefont
			{I.}~\bibnamefont {Tavernelli}},\ and\ \bibinfo {author} {\bibfnamefont
			{S.}~\bibnamefont {Filipp}},\ }\bibfield  {title} {\bibinfo {title}
		{Entanglement {G}eneration in {S}uperconducting {Q}ubits {U}sing {H}olonomic
			{O}perations},\ }\href {https://doi.org/10.1103/PhysRevApplied.11.014017}
	{\bibfield  {journal} {\bibinfo  {journal} {Phys. Rev. Applied}\ }\textbf
		{\bibinfo {volume} {11}},\ \bibinfo {pages} {014017} (\bibinfo {year}
		{2019})}\BibitemShut {NoStop}%
	\bibitem [{\citenamefont {Katz}\ \emph {et~al.}(2020)\citenamefont {Katz},
		\citenamefont {Shaham}, \citenamefont {Polzik},\ and\ \citenamefont
		{Firstenberg}}]{Katz2020}%
	\BibitemOpen
	\bibfield  {author} {\bibinfo {author} {\bibfnamefont {O.}~\bibnamefont
			{Katz}}, \bibinfo {author} {\bibfnamefont {R.}~\bibnamefont {Shaham}},
		\bibinfo {author} {\bibfnamefont {E.~S.}\ \bibnamefont {Polzik}},\ and\
		\bibinfo {author} {\bibfnamefont {O.}~\bibnamefont {Firstenberg}},\
	}\bibfield  {title} {\bibinfo {title} {Long-lived {E}ntanglement {G}eneration
			of {N}uclear {S}pins {U}sing {C}oherent {L}ight},\ }\href
	{https://doi.org/10.1103/PhysRevLett.124.043602} {\bibfield  {journal}
		{\bibinfo  {journal} {Phys. Rev. Lett.}\ }\textbf {\bibinfo {volume} {124}},\
		\bibinfo {pages} {043602} (\bibinfo {year} {2020})}\BibitemShut {NoStop}%
	\bibitem [{\citenamefont {Xi}\ \emph {et~al.}(2019)\citenamefont {Xi},
		\citenamefont {Zhang}, \citenamefont {Zheng}, \citenamefont {Li-Jost},\ and\
		\citenamefont {Fei}}]{Xi2019}%
	\BibitemOpen
	\bibfield  {author} {\bibinfo {author} {\bibfnamefont {Y.}~\bibnamefont
			{Xi}}, \bibinfo {author} {\bibfnamefont {T.}~\bibnamefont {Zhang}}, \bibinfo
		{author} {\bibfnamefont {Z.-J.}\ \bibnamefont {Zheng}}, \bibinfo {author}
		{\bibfnamefont {X.}~\bibnamefont {Li-Jost}},\ and\ \bibinfo {author}
		{\bibfnamefont {S.-M.}\ \bibnamefont {Fei}},\ }\bibfield  {title} {\bibinfo
		{title} {Converting quantum coherence to genuine multipartite entanglement
			and nonlocality},\ }\href {https://doi.org/10.1103/PhysRevA.100.022310}
	{\bibfield  {journal} {\bibinfo  {journal} {Phys. Rev. A}\ }\textbf {\bibinfo
			{volume} {100}},\ \bibinfo {pages} {022310} (\bibinfo {year}
		{2019})}\BibitemShut {NoStop}%
	\bibitem [{\citenamefont {Korzekwa}\ \emph {et~al.}(2019)\citenamefont
		{Korzekwa}, \citenamefont {Chubb},\ and\ \citenamefont
		{Tomamichel}}]{Korzekwa2019}%
	\BibitemOpen
	\bibfield  {author} {\bibinfo {author} {\bibfnamefont {K.}~\bibnamefont
			{Korzekwa}}, \bibinfo {author} {\bibfnamefont {C.~T.}\ \bibnamefont
			{Chubb}},\ and\ \bibinfo {author} {\bibfnamefont {M.}~\bibnamefont
			{Tomamichel}},\ }\bibfield  {title} {\bibinfo {title} {Avoiding
			{I}rreversibility: {E}ngineering {R}esonant {C}onversions of {Q}uantum
			{R}esources},\ }\href {https://doi.org/10.1103/PhysRevLett.122.110403}
	{\bibfield  {journal} {\bibinfo  {journal} {Phys. Rev. Lett.}\ }\textbf
		{\bibinfo {volume} {122}},\ \bibinfo {pages} {110403} (\bibinfo {year}
		{2019})}\BibitemShut {NoStop}%
	\bibitem [{\citenamefont {Piccione}\ \emph {et~al.}(2020)\citenamefont
		{Piccione}, \citenamefont {Militello}, \citenamefont {Napoli},\ and\
		\citenamefont {Bellomo}}]{Piccione2020Energy}%
	\BibitemOpen
	\bibfield  {author} {\bibinfo {author} {\bibfnamefont {N.}~\bibnamefont
			{Piccione}}, \bibinfo {author} {\bibfnamefont {B.}~\bibnamefont {Militello}},
		\bibinfo {author} {\bibfnamefont {A.}~\bibnamefont {Napoli}},\ and\ \bibinfo
		{author} {\bibfnamefont {B.}~\bibnamefont {Bellomo}},\ }\bibfield  {title}
	{\bibinfo {title} {Energy bounds for entangled states},\ }\href
	{https://doi.org/10.1103/PhysRevResearch.2.022057} {\bibfield  {journal}
		{\bibinfo  {journal} {Phys. Rev. Research}\ }\textbf {\bibinfo {volume}
			{2}},\ \bibinfo {pages} {022057(R)} (\bibinfo {year} {2020})}\BibitemShut
	{NoStop}%
	\bibitem [{\citenamefont {Ikonen}\ \emph {et~al.}(2017)\citenamefont {Ikonen},
		\citenamefont {Salmilehto},\ and\ \citenamefont {M\"ott\"onen}}]{Ikonen2017}%
	\BibitemOpen
	\bibfield  {author} {\bibinfo {author} {\bibfnamefont {J.}~\bibnamefont
			{Ikonen}}, \bibinfo {author} {\bibfnamefont {J.}~\bibnamefont {Salmilehto}},\
		and\ \bibinfo {author} {\bibfnamefont {M.}~\bibnamefont {M\"ott\"onen}},\
	}\bibfield  {title} {\bibinfo {title} {Energy-efficient quantum computing},\
	}\href {https://doi.org/10.1038/s41534-017-0015-5} {\bibfield  {journal}
		{\bibinfo  {journal} {npj Quantum Inf.}\ }\textbf {\bibinfo {volume} {3}},\
		\bibinfo {pages} {17} (\bibinfo {year} {2017})}\BibitemShut {NoStop}%
	\bibitem [{\citenamefont {Bakhshinezhad}\ \emph {et~al.}(2019)\citenamefont
		{Bakhshinezhad}, \citenamefont {Clivaz}, \citenamefont {Vitagliano},
		\citenamefont {Erker}, \citenamefont {Rezakhani}, \citenamefont {Huber},\
		and\ \citenamefont {Friis}}]{Bakhshinezhad2019}%
	\BibitemOpen
	\bibfield  {author} {\bibinfo {author} {\bibfnamefont {F.}~\bibnamefont
			{Bakhshinezhad}}, \bibinfo {author} {\bibfnamefont {F.}~\bibnamefont
			{Clivaz}}, \bibinfo {author} {\bibfnamefont {G.}~\bibnamefont {Vitagliano}},
		\bibinfo {author} {\bibfnamefont {P.}~\bibnamefont {Erker}}, \bibinfo
		{author} {\bibfnamefont {A.}~\bibnamefont {Rezakhani}}, \bibinfo {author}
		{\bibfnamefont {M.}~\bibnamefont {Huber}},\ and\ \bibinfo {author}
		{\bibfnamefont {N.}~\bibnamefont {Friis}},\ }\bibfield  {title} {\bibinfo
		{title} {Thermodynamically optimal creation of correlations},\ }\href
	{https://doi.org/10.1088/1751-8121/ab3932} {\bibfield  {journal} {\bibinfo
			{journal} {J. Phys. A Math. Theor.}\ }\textbf {\bibinfo {volume} {52}},\
		\bibinfo {pages} {465303} (\bibinfo {year} {2019})}\BibitemShut {NoStop}%
	\bibitem [{\citenamefont {Galve}\ and\ \citenamefont {Lutz}(2009)}]{Galve2009}%
	\BibitemOpen
	\bibfield  {author} {\bibinfo {author} {\bibfnamefont {F.}~\bibnamefont
			{Galve}}\ and\ \bibinfo {author} {\bibfnamefont {E.}~\bibnamefont {Lutz}},\
	}\bibfield  {title} {\bibinfo {title} {Energy cost and optimal entanglement
			production in harmonic chains},\ }\href
	{https://doi.org/10.1103/PhysRevA.79.032327} {\bibfield  {journal} {\bibinfo
			{journal} {Phys. Rev. A}\ }\textbf {\bibinfo {volume} {79}},\ \bibinfo
		{pages} {032327} (\bibinfo {year} {2009})}\BibitemShut {NoStop}%
	\bibitem [{\citenamefont {B{\'{e}}ny}\ \emph {et~al.}(2018)\citenamefont
		{B{\'{e}}ny}, \citenamefont {Chubb}, \citenamefont {Farrelly},\ and\
		\citenamefont {Osborne}}]{Beny2018}%
	\BibitemOpen
	\bibfield  {author} {\bibinfo {author} {\bibfnamefont {C.}~\bibnamefont
			{B{\'{e}}ny}}, \bibinfo {author} {\bibfnamefont {C.~T.}\ \bibnamefont
			{Chubb}}, \bibinfo {author} {\bibfnamefont {T.}~\bibnamefont {Farrelly}},\
		and\ \bibinfo {author} {\bibfnamefont {T.~J.}\ \bibnamefont {Osborne}},\
	}\bibfield  {title} {\bibinfo {title} {Energy cost of entanglement extraction
			in complex quantum systems},\ }\href
	{https://doi.org/10.1038/s41467-018-06153-w} {\bibfield  {journal} {\bibinfo
			{journal} {Nat. Commun.}\ }\textbf {\bibinfo {volume} {9}},\ \bibinfo {pages}
		{3792} (\bibinfo {year} {2018})}\BibitemShut {NoStop}%
	\bibitem [{\citenamefont {Hackl}\ and\ \citenamefont
		{Jonsson}(2019)}]{Hackl2019}%
	\BibitemOpen
	\bibfield  {author} {\bibinfo {author} {\bibfnamefont {L.}~\bibnamefont
			{Hackl}}\ and\ \bibinfo {author} {\bibfnamefont {R.~H.}\ \bibnamefont
			{Jonsson}},\ }\bibfield  {title} {\bibinfo {title} {Minimal energy cost of
			entanglement extraction},\ }\href {https://doi.org/10.22331/q-2019-07-15-165}
	{\bibfield  {journal} {\bibinfo  {journal} {{Quantum}}\ }\textbf {\bibinfo
			{volume} {3}},\ \bibinfo {pages} {165} (\bibinfo {year} {2019})}\BibitemShut
	{NoStop}%
	\bibitem [{\citenamefont {Piccione}\ \emph {et~al.}(2019)\citenamefont
		{Piccione}, \citenamefont {Militello}, \citenamefont {Napoli},\ and\
		\citenamefont {Bellomo}}]{Piccione2019Work}%
	\BibitemOpen
	\bibfield  {author} {\bibinfo {author} {\bibfnamefont {N.}~\bibnamefont
			{Piccione}}, \bibinfo {author} {\bibfnamefont {B.}~\bibnamefont {Militello}},
		\bibinfo {author} {\bibfnamefont {A.}~\bibnamefont {Napoli}},\ and\ \bibinfo
		{author} {\bibfnamefont {B.}~\bibnamefont {Bellomo}},\ }\bibfield  {title}
	{\bibinfo {title} {Simple scheme for extracting work with a single bath},\
	}\href {https://doi.org/10.1103/PhysRevE.100.032143} {\bibfield  {journal}
		{\bibinfo  {journal} {Phys. Rev. E}\ }\textbf {\bibinfo {volume} {100}},\
		\bibinfo {pages} {032143} (\bibinfo {year} {2019})}\BibitemShut {NoStop}%
	\bibitem [{\citenamefont {Bechmann-Pasquinucci}\ and\ \citenamefont
		{Peres}(2000)}]{Bechmann-Pasquinucci2000}%
	\BibitemOpen
	\bibfield  {author} {\bibinfo {author} {\bibfnamefont {H.}~\bibnamefont
			{Bechmann-Pasquinucci}}\ and\ \bibinfo {author} {\bibfnamefont
			{A.}~\bibnamefont {Peres}},\ }\bibfield  {title} {\bibinfo {title} {Quantum
			{C}ryptography with 3-{S}tate {S}ystems},\ }\href
	{https://doi.org/10.1103/PhysRevLett.85.3313} {\bibfield  {journal} {\bibinfo
			{journal} {Phys. Rev. Lett.}\ }\textbf {\bibinfo {volume} {85}},\ \bibinfo
		{pages} {3313} (\bibinfo {year} {2000})}\BibitemShut {NoStop}%
	\bibitem [{\citenamefont {Bullock}\ \emph {et~al.}(2005)\citenamefont
		{Bullock}, \citenamefont {O'Leary},\ and\ \citenamefont
		{Brennen}}]{Bullock2005}%
	\BibitemOpen
	\bibfield  {author} {\bibinfo {author} {\bibfnamefont {S.~S.}\ \bibnamefont
			{Bullock}}, \bibinfo {author} {\bibfnamefont {D.~P.}\ \bibnamefont
			{O'Leary}},\ and\ \bibinfo {author} {\bibfnamefont {G.~K.}\ \bibnamefont
			{Brennen}},\ }\bibfield  {title} {\bibinfo {title} {Asymptotically {O}ptimal
			{Q}uantum {C}ircuits for $d$-{L}evel {S}ystems},\ }\href
	{https://doi.org/10.1103/PhysRevLett.94.230502} {\bibfield  {journal}
		{\bibinfo  {journal} {Phys. Rev. Lett.}\ }\textbf {\bibinfo {volume} {94}},\
		\bibinfo {pages} {230502} (\bibinfo {year} {2005})}\BibitemShut {NoStop}%
	\bibitem [{\citenamefont {Lanyon}\ \emph {et~al.}(2009)\citenamefont {Lanyon},
		\citenamefont {Barbieri}, \citenamefont {Almeida}, \citenamefont {Jennewein},
		\citenamefont {Ralph}, \citenamefont {Resch}, \citenamefont {Pryde},
		\citenamefont {O'Brien}, \citenamefont {Gilchrist},\ and\ \citenamefont
		{White}}]{Lanyon2009}%
	\BibitemOpen
	\bibfield  {author} {\bibinfo {author} {\bibfnamefont {B.~P.}\ \bibnamefont
			{Lanyon}}, \bibinfo {author} {\bibfnamefont {M.}~\bibnamefont {Barbieri}},
		\bibinfo {author} {\bibfnamefont {M.~P.}\ \bibnamefont {Almeida}}, \bibinfo
		{author} {\bibfnamefont {T.}~\bibnamefont {Jennewein}}, \bibinfo {author}
		{\bibfnamefont {T.~C.}\ \bibnamefont {Ralph}}, \bibinfo {author}
		{\bibfnamefont {K.~J.}\ \bibnamefont {Resch}}, \bibinfo {author}
		{\bibfnamefont {G.~J.}\ \bibnamefont {Pryde}}, \bibinfo {author}
		{\bibfnamefont {J.~L.}\ \bibnamefont {O'Brien}}, \bibinfo {author}
		{\bibfnamefont {A.}~\bibnamefont {Gilchrist}},\ and\ \bibinfo {author}
		{\bibfnamefont {A.~G.}\ \bibnamefont {White}},\ }\bibfield  {title} {\bibinfo
		{title} {Simplifying quantum logic using higher-dimensional {H}ilbert
			spaces},\ }\href {https://doi.org/10.1038/nphys1150} {\bibfield  {journal}
		{\bibinfo  {journal} {Nat. Phys.}\ }\textbf {\bibinfo {volume} {5}},\
		\bibinfo {pages} {134} (\bibinfo {year} {2009})}\BibitemShut {NoStop}%
	\bibitem [{\citenamefont {Yan}\ \emph {et~al.}(2019)\citenamefont {Yan},
		\citenamefont {Zhou}, \citenamefont {Gong}, \citenamefont {Wang},\ and\
		\citenamefont {Wen}}]{Yan2019}%
	\BibitemOpen
	\bibfield  {author} {\bibinfo {author} {\bibfnamefont {X.-Y.}\ \bibnamefont
			{Yan}}, \bibinfo {author} {\bibfnamefont {N.-R.}\ \bibnamefont {Zhou}},
		\bibinfo {author} {\bibfnamefont {L.-H.}\ \bibnamefont {Gong}}, \bibinfo
		{author} {\bibfnamefont {Y.-Q.}\ \bibnamefont {Wang}},\ and\ \bibinfo
		{author} {\bibfnamefont {X.-J.}\ \bibnamefont {Wen}},\ }\bibfield  {title}
	{\bibinfo {title} {High-dimensional quantum key distribution based on qudits
			transmission with quantum {F}ourier transform},\ }\href
	{https://link.springer.com/article/10.1007/s11128-019-2368-5#citeas}
	{\bibfield  {journal} {\bibinfo  {journal} {Quantum Inf. Process.}\ }\textbf
		{\bibinfo {volume} {18}},\ \bibinfo {pages} {271} (\bibinfo {year}
		{2019})}\BibitemShut {NoStop}%
	\bibitem [{\citenamefont {Kiktenko}\ \emph {et~al.}(2020)\citenamefont
		{Kiktenko}, \citenamefont {Nikolaeva}, \citenamefont {Xu}, \citenamefont
		{Shlyapnikov},\ and\ \citenamefont {Fedorov}}]{Kiktenko2020}%
	\BibitemOpen
	\bibfield  {author} {\bibinfo {author} {\bibfnamefont {E.~O.}\ \bibnamefont
			{Kiktenko}}, \bibinfo {author} {\bibfnamefont {A.~S.}\ \bibnamefont
			{Nikolaeva}}, \bibinfo {author} {\bibfnamefont {P.}~\bibnamefont {Xu}},
		\bibinfo {author} {\bibfnamefont {G.~V.}\ \bibnamefont {Shlyapnikov}},\ and\
		\bibinfo {author} {\bibfnamefont {A.~K.}\ \bibnamefont {Fedorov}},\
	}\bibfield  {title} {\bibinfo {title} {Scalable quantum computing with qudits
			on a graph},\ }\href {https://doi.org/10.1103/PhysRevA.101.022304} {\bibfield
		{journal} {\bibinfo  {journal} {Phys. Rev. A}\ }\textbf {\bibinfo {volume}
			{101}},\ \bibinfo {pages} {022304} (\bibinfo {year} {2020})}\BibitemShut
	{NoStop}%
	\bibitem [{\citenamefont {Godfrin}\ \emph {et~al.}(2017)\citenamefont
		{Godfrin}, \citenamefont {Ferhat}, \citenamefont {Ballou}, \citenamefont
		{Klyatskaya}, \citenamefont {Ruben}, \citenamefont {Wernsdorfer},\ and\
		\citenamefont {Balestro}}]{Godfrin2017}%
	\BibitemOpen
	\bibfield  {author} {\bibinfo {author} {\bibfnamefont {C.}~\bibnamefont
			{Godfrin}}, \bibinfo {author} {\bibfnamefont {A.}~\bibnamefont {Ferhat}},
		\bibinfo {author} {\bibfnamefont {R.}~\bibnamefont {Ballou}}, \bibinfo
		{author} {\bibfnamefont {S.}~\bibnamefont {Klyatskaya}}, \bibinfo {author}
		{\bibfnamefont {M.}~\bibnamefont {Ruben}}, \bibinfo {author} {\bibfnamefont
			{W.}~\bibnamefont {Wernsdorfer}},\ and\ \bibinfo {author} {\bibfnamefont
			{F.}~\bibnamefont {Balestro}},\ }\bibfield  {title} {\bibinfo {title}
		{Operating {Q}uantum {S}tates in {S}ingle {M}agnetic {M}olecules:
			{I}mplementation of {G}rover's {Q}uantum {A}lgorithm},\ }\href
	{https://doi.org/10.1103/PhysRevLett.119.187702} {\bibfield  {journal}
		{\bibinfo  {journal} {Phys. Rev. Lett.}\ }\textbf {\bibinfo {volume} {119}},\
		\bibinfo {pages} {187702} (\bibinfo {year} {2017})}\BibitemShut {NoStop}%
	\bibitem [{\citenamefont {Blok}\ \emph {et~al.}(2021)\citenamefont {Blok},
		\citenamefont {Ramasesh}, \citenamefont {Schuster}, \citenamefont {O'Brien},
		\citenamefont {Kreikebaum}, \citenamefont {Dahlen}, \citenamefont {Morvan},
		\citenamefont {Yoshida}, \citenamefont {Yao},\ and\ \citenamefont
		{Siddiqi}}]{Blok2020}%
	\BibitemOpen
	\bibfield  {author} {\bibinfo {author} {\bibfnamefont {M.~S.}\ \bibnamefont
			{Blok}}, \bibinfo {author} {\bibfnamefont {V.~V.}\ \bibnamefont {Ramasesh}},
		\bibinfo {author} {\bibfnamefont {T.}~\bibnamefont {Schuster}}, \bibinfo
		{author} {\bibfnamefont {K.}~\bibnamefont {O'Brien}}, \bibinfo {author}
		{\bibfnamefont {J.~M.}\ \bibnamefont {Kreikebaum}}, \bibinfo {author}
		{\bibfnamefont {D.}~\bibnamefont {Dahlen}}, \bibinfo {author} {\bibfnamefont
			{A.}~\bibnamefont {Morvan}}, \bibinfo {author} {\bibfnamefont
			{B.}~\bibnamefont {Yoshida}}, \bibinfo {author} {\bibfnamefont {N.~Y.}\
			\bibnamefont {Yao}},\ and\ \bibinfo {author} {\bibfnamefont {I.}~\bibnamefont
			{Siddiqi}},\ }\bibfield  {title} {\bibinfo {title} {{Q}uantum {I}nformation
			{S}crambling on a {S}uperconducting {Q}utrit {P}rocessor},\ }\href
	{https://doi.org/10.1103/PhysRevX.11.021010} {\bibfield  {journal} {\bibinfo
			{journal} {Phys. Rev. X}\ }\textbf {\bibinfo {volume} {11}},\ \bibinfo
		{pages} {021010} (\bibinfo {year} {2021})}\BibitemShut {NoStop}%
	\bibitem [{\citenamefont {Kues}\ \emph {et~al.}(2017)\citenamefont {Kues},
		\citenamefont {Reimer}, \citenamefont {Roztocki}, \citenamefont {Cort{\'e}s},
		\citenamefont {Sciara}, \citenamefont {Wetzel}, \citenamefont {Zhang},
		\citenamefont {Cino}, \citenamefont {Chu}, \citenamefont {Little} \emph
		{et~al.}}]{Kues2017FULL}%
	\BibitemOpen
	\bibfield  {author} {\bibinfo {author} {\bibfnamefont {M.}~\bibnamefont
			{Kues}}, \bibinfo {author} {\bibfnamefont {C.}~\bibnamefont {Reimer}},
		\bibinfo {author} {\bibfnamefont {P.}~\bibnamefont {Roztocki}}, \bibinfo
		{author} {\bibfnamefont {L.~R.}\ \bibnamefont {Cort{\'e}s}}, \bibinfo
		{author} {\bibfnamefont {S.}~\bibnamefont {Sciara}}, \bibinfo {author}
		{\bibfnamefont {B.}~\bibnamefont {Wetzel}}, \bibinfo {author} {\bibfnamefont
			{Y.}~\bibnamefont {Zhang}}, \bibinfo {author} {\bibfnamefont
			{A.}~\bibnamefont {Cino}}, \bibinfo {author} {\bibfnamefont {S.~T.}\
			\bibnamefont {Chu}}, \bibinfo {author} {\bibfnamefont {B.~E.}\ \bibnamefont
			{Little}}, \emph {et~al.},\ }\bibfield  {title} {\bibinfo {title} {On-chip
			generation of high-dimensional entangled quantum states and their coherent
			control},\ }\href {https://doi.org/10.1038/nature22986} {\bibfield  {journal}
		{\bibinfo  {journal} {Nature (London)}\ }\textbf {\bibinfo {volume} {546}},\
		\bibinfo {pages} {622} (\bibinfo {year} {2017})}\BibitemShut {NoStop}%
	\bibitem [{\citenamefont {Reimer}\ \emph {et~al.}(2019)\citenamefont {Reimer},
		\citenamefont {Sciara}, \citenamefont {Roztocki}, \citenamefont {Islam},
		\citenamefont {Cort{\'e}s}, \citenamefont {Zhang}, \citenamefont {Fischer},
		\citenamefont {Loranger}, \citenamefont {Kashyap}, \citenamefont {Cino} \emph
		{et~al.}}]{Reimer2019}%
	\BibitemOpen
	\bibfield  {author} {\bibinfo {author} {\bibfnamefont {C.}~\bibnamefont
			{Reimer}}, \bibinfo {author} {\bibfnamefont {S.}~\bibnamefont {Sciara}},
		\bibinfo {author} {\bibfnamefont {P.}~\bibnamefont {Roztocki}}, \bibinfo
		{author} {\bibfnamefont {M.}~\bibnamefont {Islam}}, \bibinfo {author}
		{\bibfnamefont {L.~R.}\ \bibnamefont {Cort{\'e}s}}, \bibinfo {author}
		{\bibfnamefont {Y.}~\bibnamefont {Zhang}}, \bibinfo {author} {\bibfnamefont
			{B.}~\bibnamefont {Fischer}}, \bibinfo {author} {\bibfnamefont
			{S.}~\bibnamefont {Loranger}}, \bibinfo {author} {\bibfnamefont
			{R.}~\bibnamefont {Kashyap}}, \bibinfo {author} {\bibfnamefont
			{A.}~\bibnamefont {Cino}}, \emph {et~al.},\ }\bibfield  {title} {\bibinfo
		{title} {High-dimensional one-way quantum processing implemented on $d$-level
			cluster states},\ }\href {https://doi.org/10.1038/s41567-018-0347-x}
	{\bibfield  {journal} {\bibinfo  {journal} {Nat. Phys.}\ }\textbf {\bibinfo
			{volume} {15}},\ \bibinfo {pages} {148} (\bibinfo {year} {2019})}\BibitemShut
	{NoStop}%
	\bibitem [{\citenamefont {Imany}\ \emph {et~al.}(2019)\citenamefont {Imany},
		\citenamefont {Jaramillo-Villegas}, \citenamefont {Alshaykh}, \citenamefont
		{Lukens}, \citenamefont {Odele}, \citenamefont {Moore}, \citenamefont
		{Leaird}, \citenamefont {Qi},\ and\ \citenamefont {Weiner}}]{Imany2019}%
	\BibitemOpen
	\bibfield  {author} {\bibinfo {author} {\bibfnamefont {P.}~\bibnamefont
			{Imany}}, \bibinfo {author} {\bibfnamefont {J.~A.}\ \bibnamefont
			{Jaramillo-Villegas}}, \bibinfo {author} {\bibfnamefont {M.~S.}\ \bibnamefont
			{Alshaykh}}, \bibinfo {author} {\bibfnamefont {J.~M.}\ \bibnamefont
			{Lukens}}, \bibinfo {author} {\bibfnamefont {O.~D.}\ \bibnamefont {Odele}},
		\bibinfo {author} {\bibfnamefont {A.~J.}\ \bibnamefont {Moore}}, \bibinfo
		{author} {\bibfnamefont {D.~E.}\ \bibnamefont {Leaird}}, \bibinfo {author}
		{\bibfnamefont {M.}~\bibnamefont {Qi}},\ and\ \bibinfo {author}
		{\bibfnamefont {A.~M.}\ \bibnamefont {Weiner}},\ }\bibfield  {title}
	{\bibinfo {title} {High-dimensional optical quantum logic in large
			operational spaces},\ }\href {https://doi.org/10.1038/s41534-019-0173-8}
	{\bibfield  {journal} {\bibinfo  {journal} {npj Quantum Inf.}\ }\textbf
		{\bibinfo {volume} {5}},\ \bibinfo {pages} {59} (\bibinfo {year}
		{2019})}\BibitemShut {NoStop}%
	\bibitem [{\citenamefont {Wilmott}(2011)}]{Wilmott2011}%
	\BibitemOpen
	\bibfield  {author} {\bibinfo {author} {\bibfnamefont {C.}~\bibnamefont
			{Wilmott}},\ }\bibfield  {title} {\bibinfo {title} {On swapping the states of
			two qudits},\ }\href {https://doi.org/10.1142/S0219749911008143} {\bibfield
		{journal} {\bibinfo  {journal} {Int. J. Quantum Inf.}\ }\textbf {\bibinfo
			{volume} {09}},\ \bibinfo {pages} {1511} (\bibinfo {year}
		{2011})}\BibitemShut {NoStop}%
	\bibitem [{\citenamefont {Gao}\ \emph {et~al.}(2019)\citenamefont {Gao},
		\citenamefont {Krenn}, \citenamefont {Kysela},\ and\ \citenamefont
		{Zeilinger}}]{Gao2019}%
	\BibitemOpen
	\bibfield  {author} {\bibinfo {author} {\bibfnamefont {X.}~\bibnamefont
			{Gao}}, \bibinfo {author} {\bibfnamefont {M.}~\bibnamefont {Krenn}}, \bibinfo
		{author} {\bibfnamefont {J.}~\bibnamefont {Kysela}},\ and\ \bibinfo {author}
		{\bibfnamefont {A.}~\bibnamefont {Zeilinger}},\ }\bibfield  {title} {\bibinfo
		{title} {Arbitrary $d$-dimensional {P}auli ${X}$ gates of a flying qudit},\
	}\href {https://doi.org/10.1103/PhysRevA.99.023825} {\bibfield  {journal}
		{\bibinfo  {journal} {Phys. Rev. A}\ }\textbf {\bibinfo {volume} {99}},\
		\bibinfo {pages} {023825} (\bibinfo {year} {2019})}\BibitemShut {NoStop}%
	\bibitem [{\citenamefont {Isdrail{\u{a}}}\ \emph {et~al.}(2019)\citenamefont
		{Isdrail{\u{a}}}, \citenamefont {Kusko},\ and\ \citenamefont
		{Ionicioiu}}]{Isdrailua2019}%
	\BibitemOpen
	\bibfield  {author} {\bibinfo {author} {\bibfnamefont {T.-A.}\ \bibnamefont
			{Isdrail{\u{a}}}}, \bibinfo {author} {\bibfnamefont {C.}~\bibnamefont
			{Kusko}},\ and\ \bibinfo {author} {\bibfnamefont {R.}~\bibnamefont
			{Ionicioiu}},\ }\bibfield  {title} {\bibinfo {title} {Cyclic permutations for
			qudits in $d$ dimensions},\ }\href
	{https://doi.org/10.1038/s41598-019-42708-7} {\bibfield  {journal} {\bibinfo
			{journal} {Sci. Rep.}\ }\textbf {\bibinfo {volume} {9}},\ \bibinfo {pages}
		{6337} (\bibinfo {year} {2019})}\BibitemShut {NoStop}%
	\bibitem [{\citenamefont {Low}\ \emph {et~al.}(2020)\citenamefont {Low},
		\citenamefont {White}, \citenamefont {Cox}, \citenamefont {Day},\ and\
		\citenamefont {Senko}}]{Low2020}%
	\BibitemOpen
	\bibfield  {author} {\bibinfo {author} {\bibfnamefont {P.~J.}\ \bibnamefont
			{Low}}, \bibinfo {author} {\bibfnamefont {B.~M.}\ \bibnamefont {White}},
		\bibinfo {author} {\bibfnamefont {A.~A.}\ \bibnamefont {Cox}}, \bibinfo
		{author} {\bibfnamefont {M.~L.}\ \bibnamefont {Day}},\ and\ \bibinfo {author}
		{\bibfnamefont {C.}~\bibnamefont {Senko}},\ }\bibfield  {title} {\bibinfo
		{title} {Practical trapped-ion protocols for universal qudit-based quantum
			computing},\ }\href {https://doi.org/10.1103/PhysRevResearch.2.033128}
	{\bibfield  {journal} {\bibinfo  {journal} {Phys. Rev. Research}\ }\textbf
		{\bibinfo {volume} {2}},\ \bibinfo {pages} {033128} (\bibinfo {year}
		{2020})}\BibitemShut {NoStop}%
	\bibitem [{\citenamefont {Vidal}(2000)}]{Vidal2000}%
	\BibitemOpen
	\bibfield  {author} {\bibinfo {author} {\bibfnamefont {G.}~\bibnamefont
			{Vidal}},\ }\bibfield  {title} {\bibinfo {title} {{Entanglement monotones}},\
	}\href {https://doi.org/10.1080/09500340008244048} {\bibfield  {journal}
		{\bibinfo  {journal} {J. Mod. Opt.}\ }\textbf {\bibinfo {volume} {47}},\
		\bibinfo {pages} {355} (\bibinfo {year} {2000})}\BibitemShut {NoStop}%
	\bibitem [{\citenamefont {Plenio}\ and\ \citenamefont
		{Virmani}(2007)}]{Plenio2007}%
	\BibitemOpen
	\bibfield  {author} {\bibinfo {author} {\bibfnamefont {M.~B.}\ \bibnamefont
			{Plenio}}\ and\ \bibinfo {author} {\bibfnamefont {S.}~\bibnamefont
			{Virmani}},\ }\bibfield  {title} {\bibinfo {title} {An introduction to
			{E}ntanglement {M}easures},\ }\href
	{http://dl.acm.org/citation.cfm?id=2011706.2011707} {\bibfield  {journal}
		{\bibinfo  {journal} {Quantum Inf. Comput.}\ }\textbf {\bibinfo {volume}
			{7}},\ \bibinfo {pages} {1} (\bibinfo {year} {2007})}\BibitemShut {NoStop}%
	\bibitem [{Note1()}]{Note1}%
	\BibitemOpen
	\bibinfo {note} {\label {note1}The states which have zero projection on $\ket
		{E_0}$ can indeed be obtained by means of simple modifications of the method.
		See the relevant discussion in Appendix~\ref
		{sec:APPUnitaryOperatorConstruction}.}\BibitemShut {Stop}%
	\bibitem [{Note2()}]{Note2}%
	\BibitemOpen
	\bibinfo {note} {In Ref.~\cite {Wilmott2011} the CNOT gate is given for a
		bipartite system composed of two qubits with the same dimensions. Here, we
		have used what seemed to us the most natural extension of the CNOT gate
		operator, i.e., using system $B$ as a qudit of dimension $N_A$. Therefore, we
		apply the CNOT gate such that it involves only the first $N_A$ levels of both
		systems.}\BibitemShut {Stop}%
	\bibitem [{Note3()}]{Note3}%
	\BibitemOpen
	\bibinfo {note} {Indeed, the operator $U_B$ acts as $U_A$ on the first $N_A$
		levels of system $B$, while leaving unvaried the others.}\BibitemShut {Stop}%
	\bibitem [{\citenamefont {Breuer}\ and\ \citenamefont
		{Petruccione}(2007)}]{BookBreuer2002}%
	\BibitemOpen
	\bibfield  {author} {\bibinfo {author} {\bibfnamefont {H.-P.}\ \bibnamefont
			{Breuer}}\ and\ \bibinfo {author} {\bibfnamefont {F.}~\bibnamefont
			{Petruccione}},\ }\href@noop {} {\emph {\bibinfo {title} {The Theory of Open
				Quantum Systems}}}\ (\bibinfo  {publisher} {Oxford University Press, New
		York},\ \bibinfo {year} {2007})\BibitemShut {NoStop}%
\end{thebibliography}

%Control: production of eprint (1) enabled
%

\end{document}